%% file: ieee_transactions.tex
\documentclass[onecolumn,nofootinbib,superscriptaddress,notitlepage,10pt,tightenlines,pra,aps]{revtex4-1}

\usepackage[ colorlinks = true,
             linkcolor = blue,
             urlcolor  = blue,
             citecolor = red,
             anchorcolor = green,
]{hyperref}

\usepackage{mtunicode}
\usepackage{qittools}
\usepackage{thmtools}

\newcommand{\eqdef}{:=}

\newcommand{\ext}{\mathrm{Ext}}
\newcommand{\cond}{\mathrm{Con}}

\newcommand{\matr}{\mathrm{Mat}}

\newcommand{\lin}{\mathrm{Lin}}
\newcommand{\cb}{\mathrm{cb}}

\newcommand{\pr}[1]{\mathbf{P}\left\{ #1 \right\}}

\newcommand{\capcap}{\cap \nonscript\mskip-3mu \cdot \nonscript\mskip-3mu \cap}
\newcommand{\sigsig}{\Sigma \nonscript\mskip-3mu \cdot \nonscript\mskip-3mu \Sigma}

\input{volkherincludes.tex}

\begin{document}

\title{Quantum-proof randomness extractors via operator space theory}

\author{Mario Berta}
\email{berta@caltech.edu}

\affiliation{Institute for Quantum Information and Matter, Caltech, Pasadena, CA 91125, USA}

\author{Omar Fawzi}
\email{omar.fawzi@ens-lyon.fr}
\affiliation{Department of Computing and Mathematical Sciences, Caltech, Pasadena, CA 91125, USA}
\affiliation{ENS de Lyon, UMR 5668 LIP - CNRS - UCBL - INRIA - Université de Lyon, 69364 Lyon, France}
\affiliation{Institute for Theoretical Physics, ETH Zurich, 8093 Z\"urich, Switzerland}

\author{Volkher B.~Scholz}
\email{scholz@phys.ethz.ch}
\affiliation{Department of Physics, Ghent University, 9000 Gent, Belgium}
\affiliation{Institute for Theoretical Physics, ETH Zurich, 8093 Z\"urich, Switzerland}


\begin{abstract}
Quantum-proof randomness extractors are an important building block for classical and quantum cryptography as well as device independent randomness amplification and expansion. Furthermore they are also a useful tool in quantum Shannon theory. It is known that some extractor constructions are quantum-proof whereas others are provably not [Gavinsky {\it et al.}, STOC'07]. We argue that the theory of operator spaces offers a natural framework for studying to what extent extractors are secure against quantum adversaries: we first phrase the definition of extractors as a bounded norm condition between normed spaces, and then show that the presence of quantum adversaries corresponds to a completely bounded norm condition between operator spaces. From this we show that very high min-entropy extractors as well as extractors with small output are always (approximately) quantum-proof.

We also study a generalization of extractors called randomness condensers. We phrase the definition of condensers  as a bounded norm condition and the definition of quantum-proof condensers as a completely bounded norm condition. Seeing condensers as bipartite graphs, we then find that the bounded norm condition corresponds to an instance of a well studied combinatorial problem, called bipartite densest subgraph. Furthermore, using the characterization in terms of operator spaces, we can associate to any condenser a Bell inequality (two-player game) such that classical and quantum strategies are in one-to-one correspondence with classical and quantum attacks on the condenser. Hence, we get for every quantum-proof condenser (which includes in particular quantum-proof extractors) a Bell inequality that can not be violated by quantum mechanics.
\end{abstract}

\maketitle


\section{Introduction}\label{sec:intro}

In cryptographic protocols such as in key distribution and randomness expansion, it is often possible to guarantee that an adversary's knowledge about the secret $N$ held by honest players is bounded. The relevant quantity in many settings is the adversary's guessing probability of the secret $N$ given all his knowledge. However, the objective is usually not to create a secret that is only partly private but rather to create a (possibly smaller) secret that is almost perfectly private. The process of transforming a partly private string $N$ into one that is almost uniformly random $M$ from the adversary's point of view is called privacy amplification \cite{BBR88, BBCM95}. In order to perform privacy amplification, we apply to $N$ a function chosen at random from a set of functions $\{ f_s \}$ that has the property of being a randomness extractor. Randomness extractors are by now a standard tool used in many classical and quantum protocols. They are an essential ingredient in quantum key distribution and device independent randomness expansion protocols~\cite{Renner05,Vazirani12}, but are for example also useful in quantum Shannon theory~\cite{Dupuis12,Berta13}. For such applications, it has been only relatively recently realized~\cite{Renner05} that it is crucial to explicitly consider quantum adversaries. It is by no means obvious that a quantum adversary also satisfying the guessing probability constraint on $N$ would not be able to have additional knowledge about the output $M$. In fact, as explained below, we know of an extractor construction that becomes useless against quantum adversaries~\cite{Gavinsky07}.

We believe that in the same way as communication complexity and Bell inequalities (multi prover games), the setting of randomness extractors provides a beautiful framework for studying the power and limitations of a quantum memory compared to a classical one.
Here we argue that the theory of operator spaces, sometimes also called ``quantized functional analysis'', provides a natural arena for studying this question. This theory has already been successfully applied in the context of understanding Bell inequality violations, see~\cite{Junge10,JP11} and references therein.

This document is structured as follows. In the next two subsections we define (quantum-proof) randomness extractors (Section~\ref{sec:intro_extractor}) and condensers (Section~\ref{sec:condenser_intro}), and give a summary of the known results that are relevant to our discussion. In Section~\ref{sec:overview} we present an overview of our results (leaving out all the proofs). This is then followed by some open questions stated in Section~\ref{sec:conclusions}. For the main body of the paper, we start with basic preliminaries on the theory of normed spaces and operator spaces (Section~\ref{sec:preliminaries}). In Section~\ref{sec:extractors} we prove our results about quantum-proof extractors. The last section is devoted to the proofs of our results concerning condensers (Section~\ref{sec:condenser}). Some technical arguments are deferred to the Appendices (Appendix~\ref{app:missing}--\ref{app:intersection}).


\subsection{Randomness extractors}\label{sec:intro_extractor}

Extractors map a weakly random system into (almost) uniform random bits, with the help of perfectly random bits called the seed. We use $N=2^{n}$ to denote the input system (consisting of strings of \(n\) bits), $M=2^{m}$ (bit-strings of length \(m\)) to denote the output system, and $D=2^{d}$ (\(d\) bits) to denote the seed system. Note that in a slight abuse of notation, we use the same letters for the actual system as well as its dimension as a linear vector space. An extractor is then a family of functions $\{f_1, \dots, f_{D}\}$ with $f_s : N \to M$ satisfying the following property. For any random variable on $N$ with
\begin{align}
H_{\min}(N)\eqdef-\log p_{\mathrm{guess}}(N)\geq k\ ,
\end{align}
where $p_{\mathrm{guess}}(N)$ denotes the maximal probability of guessing the input, and an independent and uniform seed $U$, the random variable $f_U(N)$ has a distribution which is $\eps$-close in total variation distance to the uniform distribution $\upsilon_M$ on $M$. For us, it is more convenient to state the definition in terms of probability distributions. For this we associate to the functions $f_s$ an $M\times N$ matrix $F_s$ where the entry $(y,x)$ is equal to one if $f_s(x) = y$ and zero otherwise. With this notation, we have for any probability distribution $P_{N}$ of a random variable on $N$, $F_s(P_{N})$ is the distribution of $f_s(N)$. That is, a $(k,\eps)$-extractor satisfies the following property. For all input probability distributions $P_{N}$ with $H_{\min}(N)_{P}\geq k$,
\begin{align}\label{eq:intro_weakmin}
\left\|\frac{1}{D}\cdot\sum_{s=1}^{D}F_s(P_N)-\upsilon_{M}\right\|_{\ell^1_M}\leq\eps\ .
\end{align}
This definition is also referred to as weak extractors. An important special case of extractors are strong extractors~\cite{Nissan96}, for which the seed is part of the output, i.e., the output space has the form $M = D \times M'$ and $f_s(x)=(s,f'_s(x))$ for some function $f'_s:N\to M'$ (with $M'=2^{m'}$). This means that the invested randomness, the seed $D$, is not used up and can safely be reused later. Alternatively, in a cryptographic context, the seed $D$ can be published without compromising the security. The condition~\eqref{eq:intro_weakmin} then reads as
\begin{align}\label{eq:intro_strongmin}
\frac{1}{D}\cdot\sum_{s=1}^{D}\left\|F_s'(P_N)-\upsilon_{M'}\right\|_{\ell^1_M}\leq\eps\ .
\end{align}
Such objects are needed for privacy amplification, since the eavesdropper is allowed to know which function is applied.

We now briefly discuss the parameters for which strong extractors exists. Typically we want to maximize the output length $m$ and minimize the seed length $d$. Radhakrishnan and Ta-Shma~\cite{Radhakrishnan00} show that every strong $(k,\eps)$-extractor necessarily has
\begin{align}\label{eq:extractor_converse}
m\leq k-2\log(1/\eps)+O(1)\quad\mathrm{and}\quad d\geq\log(n-k)+2\log(1/\eps)-O(1)\ .
\end{align}
Using the probabilistic method one can show that random functions achieves these bounds up to constants~\cite{Sipser88,Radhakrishnan00}. There exists a strong $(k,\eps)$-extractor with
\begin{align}\label{eq:extractor_prob}
m=k-2\log(1/\eps)-O(1)\quad\mathrm{and}\quad d=\log(n-k)+2\log(1/\eps)+O(1)\ .
\end{align}
Probabilistic constructions are interesting, but for applications we usually need explicit extractors. Starting with the work by Nisan and Ta-Shma~\cite{Nisan:1999kn} and followed by Trevisan's breakthrough result~\cite{Trevisan99} there has been a lot of progress in this direction, and there are now many constructions that almost achieve the converse bounds in~\eqref{eq:extractor_converse} (see the review articles~\cite{Shaltiel02,Vadhan07}). 

For applications in classical and quantum cryptography (see, e.g., \cite{Renner05,Koenig12}), for constructing device independent randomness amplification and expansion schemes (see, e.g., \cite{Chung14,Miller14,Brandao13,Coudron13}), and for applications in quantum Shannon theory (see, e.g., \cite{Dupuis12,Berta13}) it is important to find out if extractor constructions also work when the input source is correlated to another (possibly quantum) system $Q$. That is, we would like that for all classical-quantum input density matrices $\rho_{QN}$ with conditional min-entropy
\begin{align}
H_{\min}(N|Q)_{\rho}\eqdef-\log p_{\mathrm{guess}}(N|Q)_{\rho}\geq k\ ,
\end{align}
where $p_{\mathrm{guess}}(N|Q)$ denotes the maximal probability of guessing $N$ given $Q$, the output is uniform and independent of $Q$,\footnote{Other notions for weaker quantum adversaires have also been discussed in the literature, e.g., in the bounded storage model (see~\cite[Section 1]{De12} for a detailed overview).}
\begin{align}\label{eq:q_extractor}
\left\|\frac{1}{D}\sum_{s=1}^{D} (\id_Q \otimes F_s)(\rho_{QN})-\rho_{Q}\otimes\upsilon_{M}\right\|_{1}\leq\eps\ ,
\end{align}
and $\|\cdot\|_{1}$ denotes the trace norm (the quantum extension of the $\ell^1$-norm). Similarly, the corresponding condition for quantum-proof strong extractors reads
\begin{align}\label{eq:qs_extractor}
\frac{1}{D}\sum_{s=1}^{D} \left\|(\id_Q \otimes F'_s)(\rho_{QN})-\rho_{Q}\otimes\upsilon_{M'}\right\|_{1}\leq\eps\ .
\end{align}
K\"onig and Terhal~\cite[Proposition 1]{Koenig08} observed that if we restrict the system $Q$ to be classical with respect to some basis $\{\ket{e}\}_{e\in Q}$ then every $(k,\eps)$-extractor as in~\eqref{eq:intro_weakmin} is also a $\big(k+\log(1/\eps),2\eps\big)$-extractor in the sense of~\eqref{eq:q_extractor} (and the analogue statement for strong extractors is a special case of this). That is, even when the input source is correlated to a classical system $Q$, every extractor construction still works (nearly) equally well for extracting randomness. However, if $Q$ is quantum no such generic reduction is known and extractor constructions that also work for quantum $Q$ are called quantum-proof.\footnote{Note that the dimension of $Q$ is unbounded and that it is a priori unclear if there exist any extractor constructions that are quantum-proof (even with arbitrarily worse parameters). Furthermore, and in contrast to some claims in the literature, we believe that the question to what extent extractors are quantum-proof is already interesting for weak extractors. In particular, if weak extractors were perfectly quantum-proof then strong extractors would be perfectly quantum-proof as well (and we know that the latter is wrong~\cite{Gavinsky07}).} Examples of (approximately) quantum-proof extractors include:
\begin{itemize}
\item Spectral $(k,\eps)$-extractors are quantum-proof $(k,2\sqrt{\eps})$-extractors~\cite[Theorem 4]{Berta14}. This includes in particular two-universal hashing~\cite{Renner05,Tomamichel11}, two-wise independent permutations~\cite{Szehr11}, as well as sample and hash based constructions~\cite{Koenig11}.
\item One bit output strong $(k,\eps)$-extractors are quantum-proof strong $(k+\log(1/\eps),3\sqrt{\eps})$-extractors~\cite[Theorem 1]{Koenig08}.
\item  Strong $(k,\eps)$-extractors constructed along Trevisan's ideas~\cite{Trevisan99} are quantum-proof strong $\big(k+\log(1/\eps),3\sqrt{\eps}\big)$-extractors~\cite[Theorem 4.6]{De12} (see also~\cite{Ben-Aroya12}).
\end{itemize}
We emphasize that all these stability results are specifically tailored proofs that make use of the structure of the particular extractor constructions. In contrast to these findings it was shown by Gavinsky {\it et al.}~\cite[Theorem 1]{Gavinsky07} that there exists a valid (though contrived) strong extractor for which the decrease in the quality of the output randomness has to be at least $\eps\mapsto \Omega(m'\eps)$.\footnote{Since the quality of the output randomness of Gavinsky {\it et al.}'s construction is bad to start with, the decrease $\eps\mapsto \Omega(m'\eps)$ for quantum $Q$ already makes the extractor fail completely in this case.} As put forward by Ta-Shma~\cite[Slide 84]{TaShma13}, this then raises the question if the separation found by Gavinsky {\it et al.} is maximal, that is:
\begin{center}
\begin{minipage}{0.92\textwidth}
{\it Is every $(k,\eps)$-extractor a quantum-proof $\big(O(k+\log(1/\eps)),O(m\sqrt{\eps})\big)$-extractor or does there exists an extractor that is not quantum-proof with a large separation, say $\eps\mapsto (2^{m}\eps)^{\Omega(1)}$?}
\end{minipage}
\end{center}
We note that such a stability result would make every extractor with reasonable parameters (approximately) quantum-proof.


\subsection{Randomness condensers}\label{sec:condenser_intro}

In the general theory of pseudorandomness one interesting generalization of extractors are condensers. These objects were first defined in~\cite{RR99,RSW00} as an intermediate step in order to build extractors. For condensers the output is not necessarily close to the uniform distribution (as it is the case for extractors), but only close to a distribution with high min-entropy $k'$. More precisely, a (weak) $(k\to_{\eps}k')$-condenser is family of functions $\{f_1, \dots, f_{D}\}$ with $f_s:N\to M$ such that for all input probability distributions $P_{N}$ with $H_{\min}(N)_{P}\geq k$
\begin{align}\label{eq:w_condenser}
H_{\min}^{\eps}(M)_{\frac{1}{D}\sum_{s=1}^{D}F_{s}(P)}\geq k'\ ,
\end{align}
with the smooth min-entropy
\begin{align}
H_{\min}^{\eps}(N)_{P}\eqdef\sup_{R}H_{\min}(N)_{R}\ ,
\end{align}
and the supremum over all probability distributions $R_{N}$ such that $\|R_{N}-P_{N}\|_{\ell_1}\leq\eps$. Observe that when $k'=\log M$, this is exactly the condition for being a (weak) $(k,\eps)$-extractor. The reason this definition is called condenser is because we want constructions with $M < N$ so that the entropy is condensed in a smaller probability space. For the special case of strong condensers the output space has the form $M = D \times M'$ and $f_s(x)=(s,f'_s(x))$ for some function $f'_s:N\to M'$, and the condition~\eqref{eq:w_condenser} then reads
\begin{align}
H_{\min}^{\eps}(M'D)_{\frac{1}{D}\sum_{s=1}^{D}F'_{s}(P)\otimes\proj{s}}\geq k'\ .
\end{align}
A (weak) quantum-proof condenser is as follows. For all classical-quantum input density matrices $\rho_{QN}$ with conditional min-entropy $H_{\min}(N|Q)_{\rho}\geq k$, the output should be close to a distribution with high conditional min-entropy $k'$
\begin{align}\label{eq:qcondenser}
H_{\min}^{\eps}(M|Q)_{\frac{1}{D}\sum_{s=1}^{D}(\id\otimes F_{s})(\rho)}\geq k'\ ,
\end{align}
with the smooth conditional min-entropy
\begin{align}
H_{\min}^{\eps}(N|Q)_{\rho}\eqdef\sup_{\sigma}H_{\min}(N|Q)_{\sigma}\ ,
\end{align}
and the supremum over all classical-quantum density matrices $\sigma_{QN}$ such that $\|\sigma_{QN}-\rho_{QN}\|_{1}\leq\eps$. Similarly, the corresponding condition for quantum-proof strong condensers reads
\begin{align}
H_{\min}^{\eps}(M'D|Q)_{\frac{1}{D}\sum_{s=1}^{D}(\id\otimes F'_{s})(\rho)\otimes\proj{s}}\geq k'\ .
\end{align}
We note that the works~\cite{Koenig11,Dupuis13,Wullschleger14} can be understood as results about quantum-proof condensers. One reason why we would like to understand to what extent condensers are quantum-proof is that the best known explicit extractor constructions are built from condensers (see the review article~\cite{Vadhan07}). 

As an extractor is a special case of a condenser with full output entropy $k' = m$, one can understand the construction of Gavinsky {\it et al.}~\cite[Theorem 1]{Gavinsky07} also as a valid randomness condenser that is not quantum-proof. But a condenser has the output entropy as an additional parameter, so it is natural to ask whether this construction is a quantum-proof condenser with slightly worse parameters, e.g., $k + c'' \to_{c \eps} k' - c'$ for some constants $c, c'$ and $c''$. Note that when $c' > 0$, this does not correspond to an extractor condition anymore, as the output is only required to have large smooth min-entropy. It turns out that even if we relax the condition on the output min-entropy slightly, the condenser is still not quantum-proof. The reason is that the quantum attack given in~\cite{Gavinsky07} allows the adversary to correctly guess $\gamma$ bits of the output with a memory of size $O(\gamma \log n)$. Thus the smooth min-entropy of the output can be at most roughly $k' - \gamma$ if the input min-entropy is $n - \Omega(\gamma \log n)$.


\section{Overview of results}\label{sec:overview}

Here we state our results but leave out all the proofs, which can be found in Section~\ref{sec:extractors} and Section~\ref{sec:condenser}.

\subsection{Extractors}\label{sec:extr_results}

A linear vector space which is equipped with a norm is called a normed space (we restrict ourselves to finite-dimensional spaces). Special examples are the linear space $\C^N$, which can be equipped with the $\ell^1_N$-norm, the sum of all absolute values of vector entries, or the $\ell^\infty_N$-norm, the largest absolute value of vector entries. Both norms are useful for studying extractors, as the first norm encodes the normalization constraint (the inputs are probability distributions), while the second is just the exponential of the min-entropy. Linear maps between normed spaces are naturally equipped with norms, defined as the maximum norm of any output, given that the input had bounded norm. Of course, the norms on the input and the output spaces can be different. Our first result is that the extractor condition~\eqref{eq:intro_weakmin} can be rewritten as a condition on the norm of a linear map. In the expression~\eqref{eq:intro_weakmin}, observe that
\begin{align}
P_{N} \mapsto \frac{1}{D} \sum_{s=1}^{D} F_s(P_N)
\end{align}
is a linear map. In addition, as $P_N$ is a probability distribution, we can write $\nu_M = (\mathbf{1}^{T} P_N) \nu_M$, where $\mathbf{1}^T = (1, \dots ,1) \in \bR^N$ is the vector with all ones. As a result,
\begin{align}
P_{N} \mapsto \frac{1}{D} \sum_{s=1}^{D} F_s(P_N) - (\mathbf{1}^{T} P_N) \nu_M
\end{align}
is a linear function in the input distribution. We can associate to an extractor $\ext=\{f_s\}_{s\in D}$ a linear map from $N$ to $M$ given by
\begin{align}\label{eq:delta_map}
\Delta[\ext](P_N)\eqdef\frac{1}{D}\sum_{s=1}^{D}F_s(P_{N})- (\mathbf{1}^{T} P_{N})\nu_{M}\ .
\end{align}
Using this notation, the extractor condition can be written as follows: for all distributions $P_{N}$ with $\|P_{N}\|_{\ell^1_N}=1$ and $H_{\min}(N)_{P}\geq k$, $\|\Delta[\ext](P)\|_{\ell^1_M}\leq\eps$. In order to capture the input constraints on $P_{N}$ we now define for $k\in(0,\infty)$ the $\cap\{2^{k}\ell^{\infty}_N,\ell^{1}_N\}$-norm as
\begin{align}
\|P_{N}\|_{\cap\{2^{k}\ell^\infty_N,\ell^1_N\}}\eqdef\max\Big\{2^{k}\|P_{N}\|_{\ell^\infty_N},\|P_{N}\|_{\ell^1_N}\Big\}\ .
\end{align}
We then get the bounded norm
\begin{align}
\left\|\Delta[\ext]:\cap\{2^{k}\ell^{\infty}_N,\ell^1_N\}\to\ell^1_M\right\|\equiv\|\Delta[\ext]\|_{\cap\{2^{k}\ell^\infty_N,\ell^1_N\}\to \ell^1_M}\eqdef\sup_{\|P\|_{\cap\{2^{k}\ell^\infty,\ell^1\}}\leq1}\|\Delta[\ext](P_N)\|_{\ell^1_M}\ ,
\end{align}
that gives the following alternative characterization for extractors.

\begin{restatable}{proposition}{extToNorm}\label{prop:ext-to-norm}
Let $\ext = \{f_s\}_{s\in D}$ and $\Delta[\ext]$ as defined in~\eqref{eq:delta_map}. Then, we have
\begin{align}
\left\|\Delta[\ext]:\cap\{2^{k}\ell^\infty_N,\ell^1_N\}\to\ell^1_M\right\|\leq\eps\quad&\Rightarrow\quad\ext\,\mathrm{is}\,\mathrm{a}\,(k,\eps)\mbox{-}\,\mathrm{extractor}\label{eq:extractor_1}\\
\ext\,\mathrm{is}\,\mathrm{a}\,(k-1,\eps)\mbox{-}\,\mathrm{extractor}\quad&\Rightarrow\quad\left\|\Delta[\ext]:\cap\{2^k\ell^\infty_N,\ell^1_N\}\to\ell^1_M\right\|\leq4\eps\ .\label{eq:extractor_2}
\end{align}
\end{restatable}

Note that strong extractors are covered by this as a special case. To make this explicit we associate to a strong extractor $\ext=\{f_s\}_{s\in D}$ the linear map
\begin{align}\label{eq:delta_mapstrong}
\Delta[\ext]_{S}(P_N)\eqdef\frac{1}{D}\sum_{s=1}^{D}F_s(P_{N})\otimes\proj{s}_{D}-(\mathbf{1}^{T} P_{N})\nu_{M'}\otimes\nu_{D}\ ,
\end{align}
and the corresponding statement can then be read off by replacing $\Delta[\ext]$ with $\Delta[\ext]_{S}$ in Proposition~\ref{prop:ext-to-norm}.

The theory of normed spaces is often convenient in classical probability theory. However, if the systems of interests are in addition correlated with quantum systems, we have more structure available. A natural norm on quantum systems is the $\infty$-norm, the largest singular value. Hence we start by modeling a classical system as a normed space, and a quantum system as complex-valued matrices equipped with the $\infty$-norm. If we allow for correlations between the two, we have to define norms on their tensor product, fulfilling reasonable compatibility requirements. The framework of operator spaces axiomatizes such scenarios: an operator space is a normed space equipped with a sequence of norms describing possible correlations to quantum systems. If we now study linear maps between normed spaces, we can naturally consider these maps to be maps between operator spaces by letting them act trivially on the quantum part. Of course, the norm of the linear maps might change, since we now also allow for correlations to the quantum part (at the input as well as at the output). The associated norm, defined as the supremum with respect to quantum systems of any dimension, is called the completely bounded norm, or just $\cb$-norm.

From this discussion, it is reasonable to expect that the property of being quantum-proof can be modeled as a $\cb$-norm constraint. Indeed, we have that there exists operator spaces $\mathds{L}_N^{\infty}$ and $\mathds{L}^1_N$ defined on the normed spaces described in Proposition~\ref{prop:ext-to-norm} such that the following completely bounded norm
\begin{align}
\left\|\Delta[\ext]:\capcap\left\{2^k\mathds{L}_N^{\infty},\mathds{L}_N^1\right\}\to\mathds{L}^1_M\right\|_{\cb}\equiv\|\Delta[\ext]\|_{\capcap\{2^k\mathds{L}_N^{\infty},\mathds{L}_N^1\}\to\mathds{L}^1_M}
\end{align}
captures the property of being quantum-proof. Note that we use the notation $\capcap$ instead of $\cap$. The reason is that there is a natural operator space extension associated with the $\cap\{2^k\ell^\infty_N,\ell^1_N\}$-norm for which one would use the same name, but the operator space we consider here is slightly different.

\begin{restatable}{theorem}{extToNormQuantum}\label{thm:ext-to-norm-quantum}
Let $\ext=\{f_s\}_{s\in D}$ and $\Delta[\ext]$ as defined in~\eqref{eq:delta_map}. Then, we have
\begin{align}
\left\|\Delta[\ext]:\capcap\left\{2^k\mathds{L}_N^{\infty},\mathds{L}_N^1\right\}\to\mathds{L}^1_M\right\|_{\cb} \leq \eps\quad&\Rightarrow\quad\ext\,\mathrm{is}\,\mathrm{a}\,\mathrm{quantum}\mbox{-}\mathrm{proof}\,(k,2\eps)\mbox{-}\,\mathrm{extractor}\label{eq:quantum_ext1}\\
\ext\,\mathrm{is}\,\mathrm{a}\,\mathrm{quantum}\mbox{-}\,\mathrm{proof}\,(k-1,\eps)\mbox{-}\mathrm{extractor}\quad&\Rightarrow\quad\left\|\Delta[\ext]:\capcap\left\{2^k\mathds{L}_N^{\infty},\mathds{L}_N^1\right\}\to\mathds{L}^1_M\right\|_{\cb} \leq 8\eps\ .\label{eq:quantum_ext2}
\end{align}
\end{restatable}

Again the special case of strong extractors just follows by replacing $\Delta[\ext]$ with $\Delta[\ext]_{S}$ in Theorem~\ref{thm:ext-to-norm-quantum}.


\subsubsection*{Applications}

We conclude that the ratio between the bounded norm $\cap\{2^{k}\ell^\infty_N,\ell^1_N\} \to \ell^1_M$ and its completely bounded extension in Theorem~\ref{thm:ext-to-norm-quantum} can be used to quantify to what extent extractors are quantum-proof. This type of comparison is of course very well studied in the theory of operator spaces. As a first straightforward application, we can use dimension dependent bounds for the maximal possible ratios between the completely bounded norm and the bounded norm,
\begin{align}
\|\cdot\|_{\capcap\{2^k\mathds{L}_N^{\infty},\mathds{L}_N^1\}\to\mathds{L}^1_M}\leq\sqrt{2^{m}}\|\cdot\|_{\cap\{2^{k}\ell^\infty_N,\ell^1_N\}\to\ell^1_M}\ ,
\end{align}
and with Proposition~\ref{prop:ext-to-norm} and Theorem~\ref{thm:ext-to-norm-quantum} we find that every $(k,\eps)$-extractor is a quantum-proof $(k+1,8\sqrt{2^{m}}\eps)$-extractor.\footnote{We should point out that showing a similar bound, where the quantum error is upper bounded by $2^m$ multiplied by the classical error, can be obtained with a basic argument (not making use of any operator space theory).} If $M$ is small this bound is non-trivial for weak extractors, but for strong extractors (for which the seed $D$ is part of the output $M=D\times M'$) the bound becomes useless. However, using an operator version of the Cauchy-Schwarz inequality due to Haagerup we find the following bound for strong extractors.

\begin{restatable}{theorem}{smallerrorent}\label{thm:small-error-ent}
Let $\ext=\{f_s\}_{s\in D}$, and $\Delta[\ext]_{S}$ as defined in~\eqref{eq:delta_mapstrong}. Then, we have
\begin{align}
\left\|\Delta[\ext]_{S}:\cap\left\{2^{k}\ell^\infty_N,\ell^1_N\right\}\to\ell^1_M\right\|\leq\eps\,\Rightarrow\,\left\|\Delta[\ext]_{S}:\capcap\left\{2^{k+\log(1/\eps)}\mathds{L}_N^{\infty},\mathds{L}_N^1\right\}\to\mathds{L}^1_M\right\|_{\cb}\leq4\sqrt{2^{m'}}\sqrt{2\eps}\ .
\end{align}
\end{restatable}

By Proposition~\ref{prop:ext-to-norm} every strong $(k,\eps)$-extractor becomes a quantum-proof strong $(k+1+\log(1/(4\eps)),16\sqrt{2^{m'}}\sqrt{2\eps})$-extractor. This extends the result of K\"onig and Terhal for $m'=1$~\cite[Theorem 1]{Koenig08} to arbitrary output sizes. The bound is useful as long as the output size is small but does of course not match Ta-Shma's conjecture that asks for an error dependence of $O(m'\sqrt{\eps})$. We note that in independent work, \cite{CLW14} observed that one can use the quantum XOR lemma of~\cite{KK12} to obtain a result slightly weaker than Theorem~\ref{thm:small-error-ent}, namely the error of the extractor against quantum adversaries can be bounded by $O(2^{m'} \sqrt{\eps})$. 

We can also analyze the simpler bounded norm $2^{k}\ell^\infty_N\to\ell^1_M$ instead of $\cap\{2^{k}\ell^\infty_N,\ell^1_N\}\to \ell^1_M$. Grothendieck's inequality~\cite[Corollary 14.2]{Pis12} then shows that the ratio between the bounded norm $2^{k}\ell^\infty_N\to\ell^1_M$ and its completely bounded extension is at most Grothendieck's constant $K_{G}<1.8$:
\begin{align}
\|\cdot\|_{2^{k}\mathds{L}^\infty_N\to\mathds{L}^1_M}\leq K_{G}\|\cdot\|_{2^{k}\ell^\infty_N\to\ell^1_M}\ .
\end{align}
This gives the following bound.\footnote{Interestingly this also implies that extractors for non-normalized inputs are automatically quantum-proof (e.g., spectral extractors~\cite{Berta14}), and hence that the $\ell^1_N$-normalization condition on the input is crucial for studying to what extent extractors are quantum-proof.}

\begin{restatable}{theorem}{highMinEntExt}\label{thm:high-min-ent-ext}
Let $\ext = \{f_s\}_{s\in D}$ and $\Delta[\ext]$ as defined in~\eqref{eq:delta_map}. Then, we have
\begin{align}
\left\|\Delta[\ext]:\cap\left\{2^{k}\ell^\infty_N,\ell^1_N\right\}\to\ell^1_M\right\|\leq\eps\,\Rightarrow\,\left\|\Delta[\ext]:\capcap\left\{2^k\mathds{L}_N^{\infty},\mathds{L}_N^1\right\}\to\mathds{L}^1_M\right\|_{\cb}\leq K_{G}2^{n-k}\eps\ .
\end{align}
\end{restatable}

Hence we get by Proposition~\ref{prop:ext-to-norm} and Theorem~\ref{prop:ext-to-norm} that every $(k,\eps)$-extractor is a quantum-proof $(k+1,8K_{G}2^{n-k}\eps)$-extractor. This applies to weak and strong extractors equally since the statement is independent of the output size. So in particular if the input min-entropy is very high, e.g., $k\geq n-\log n$, we get a quantum-proof $(k+1,8K_{G}n\eps)$-extractor. Hence, very high min-entropy extractors are always (approximately) quantum-proof.\footnote{This result tightly matches with that the extractor construction of Gavinsky {\it et al.}~\cite{Gavinsky07}. In fact, their construction is an extractor for $k \geq n - n^{c}$ with error $\eps = n^{-c'}$ for some constants $c,c'$ and it fails to be quantum-proof even if $k=n-O(\log n)$ with constant $\eps$. Theorem~\ref{thm:high-min-ent-ext} says is that if the error in the classical case was a slightly smaller, for example super-polynomially small, then the extractor would have been quantum-proof.} We emphasize again that Theorem~\ref{thm:small-error-ent} and Theorem~\ref{thm:high-min-ent-ext} do only make use of the definition of extractors and not of any potential structure of extractor constructions. Finally, we note that after this paper first appeared online, we were able to prove our results about quantum-proof extractors without the operator space theory language and only using semidefinite programming~\cite{BFS15_tqc,berta16}.

\begin{table}[ht]
\centering
\begin{tabular}{l|ll}
Strong $(k,\eps)$-extractor & Quantum-proof & with parameters\\
\hline\hline
Probabilistic constructions & ? & \\
\hline
Spectral (e.g., two-universal hashing) & $\checkmark$~~~\cite[Thm.~4]{Berta14} & $(k,c\sqrt{\eps})$\\
\hline
Trevisan based & $\checkmark$~~~\cite[Thm.~4.6]{De12} & $(k+\log(1/\eps),c\sqrt{\eps})$\\
\hline
One-bit output & $\checkmark$~~~\cite[Thm.~1]{Koenig08} &$(k+\log(1/\eps),c\sqrt{\eps})$\\
\hline
Small output & $\checkmark$~~~[Thm.~\ref{thm:small-error-ent}] & $(k+\log(1/\eps),c\sqrt{2^m}\sqrt{2\eps})$\\
\hline
High entropy & $\checkmark$~~~[Thm.~\ref{thm:high-min-ent-ext}] & $(k+1,c2^{n-k}\eps)$\\
\hline
\end{tabular}
\caption{Stability results for strong extractors: input $N=2^{n}$, output $M=2^{m}$, seed $D=2^{d}$, min-entropy $k$, error parameter $\eps$, and $c$ represents possibly different constants.
}
\label{table:stability}
\end{table}


\subsection{Condensers}\label{sec:cond_results}

The framework of normed spaces and operator spaces also allows to analyze to what extent condensers are quantum-proof. Analogously as for extractors, we associate to a condenser $\cond=\{f_{s}\}_{s\in D}$ a linear map from $N$ to $M$ given by
\begin{align}\label{eq:cond_map}
[\cond](P)\eqdef\frac{1}{D}\cdot\sum_{s=1}^{D}F_{s}(P_{N})\ .
\end{align}
The input constraint for condenser is the same as for extractors and in order to characterize the output constraint~\eqref{eq:qcondenser} we define the norm
\begin{align}
\|Q_{M}\|_{\Sigma\{2^{k'}\ell^\infty_M,\ell^1_M\}}\eqdef\inf\left\{2^{k'}\|Q_{1}\|_{\ell^\infty_M}+\|Q_{2}\|_{\ell^1_M}:Q_{1}+Q_{2}=Q_{M}\right\}\ .
\end{align}
The bounded norm
\begin{align}
\left\|[\cond]:\cap\left\{2^{k}\ell^\infty_N,\ell^1_N\right\}\to\Sigma\left\{2^{k'}\ell^\infty_M,\ell^1_M\right\} \right\|&\equiv\|[\cond]\|_{\cap\{2^{k}\ell^\infty_N,\ell^1_N\}\to\Sigma\{2^{k'}\ell^\infty_M,\ell^1_M\}}\\
&\eqdef\sup_{\|P\|_{\cap\{2^{k}\ell^\infty,\ell^1\}\leq1}}\|[\cond](P_N)\|_{\Sigma\{2^{k'}\ell^\infty_M,\ell^1_M\}}
\end{align}
then gives the following norm characterization for condensers.

\begin{restatable}{proposition}{condToNorm}\label{prop:cond-to-norm}
Let $\cond=\{f_{s}\}_{s\in D}$ and $[\cond]$ as defined in~\eqref{eq:cond_map}. Then, we have
\begin{align}
\|[\cond]:\cap\{2^{k}\ell^\infty_N,\ell^1_N\}\to\Sigma\{2^{k'}\ell^\infty_M,\ell^1_M\} \|\leq\eps\quad&\Rightarrow\quad\cond\,\mathrm{is}\,\mathrm{a}\,\left(k\to_{\eps}k'+\log(1/\eps)\right)\mbox{-}\,\mathrm{condenser}\label{eq:condenser_1}\\
\cond\,\mathrm{is}\,\mathrm{a}\,\left(k-1\to_{\eps}k'+\log(1/\eps)\right)\mbox{-}\,\mathrm{condenser}\quad&\Rightarrow\quad\|[\cond]:\cap\{2^{k}\ell^\infty_N,\ell^1_N\} \to \Sigma\{2^{k'}\ell^\infty_M,\ell^1_M\}\|\leq8\eps\ .\label{eq:condenser_2}
\end{align}
\end{restatable}

Note that strong condensers are covered by this as a special case. To make this explicit we associate to a strong condenser $\cond=\{f_s\}_{s\in D}$ the linear map
\begin{align}\label{eq:cond_strongmap}
[\cond]_{S}(P_N)\eqdef\frac{1}{D}\cdot\sum_{s=1}^{D}F_{s}(P_{N})\otimes\proj{s}_{D}\ ,
\end{align}
and the corresponding statement can then be read off by replacing $[\cond]$ with $[\cond]_{S}$ in Proposition~\ref{prop:cond-to-norm}. Again we have that there exist operator spaces $\mathds{L}_N^{\infty}$ and $\mathds{L}^1_N$ defined on the normed spaces described in Proposition~\ref{prop:cond-to-norm} such that the following completely bounded norm
\begin{align}
\left\|[\cond]:\capcap\left\{2^{k}\mathds{L}_N^\infty,\mathds{L}_N^1\right\}\to\sigsig\left\{2^{k'}\mathds{L}_M^\infty,\mathds{L}^1_M\right\}\right\|_\cb\equiv\left\|[\cond]\right\|_{\capcap\{2^{k}\mathds{L}_N^\infty,\mathds{L}_N^1\}\to\sigsig\{2^{k'}\mathds{L}_M^\infty,\mathds{L}^1_M\}}
\end{align}
captures the property of being quantum-proof.

\begin{restatable}{theorem}{condToNormQuantum}\label{thm:cond-to-norm-quantum}
Let $\cond=\{f_{s}\}_{s\in D}$ and $[\cond]$ as defined in~\eqref{eq:cond_map}. Then, we have
\begin{align}
\left\|[\cond]:\capcap\left\{2^{k}\mathds{L}_N^\infty,\mathds{L}_N^1\right\}\to\sigsig\left\{2^{k'}\mathds{L}_M^\infty,\mathds{L}^1_M\right\}\right\|_\cb\leq\frac{\eps}{4}\,\Rightarrow\,\cond\,\mathrm{quantum}\mbox{-}\mathrm{proof}\left(k\to_{\eps}k'+\log(1/\eps)\right)\mbox{-}\,\mathrm{condenser}\label{eq:quantumcondenser1}
\end{align}
\begin{align}
\cond\,\mathrm{quantum}\mbox{-}\mathrm{proof}\,(k-1\to_{\eps}k'+\log(1/\eps))\mbox{-}\,\mathrm{condenser}\Rightarrow\,\left\|[\cond]:\capcap\left\{2^{k}\mathds{L}_N^\infty,\mathds{L}_N^1\right\}\to\sigsig\left\{2^{k'}\mathds{L}_M^\infty,\mathds{L}^1_M\right\}\right\|_\cb\leq8\eps\ .\label{eq:quantumcondenser2}
\end{align}
\end{restatable}

The special case of strong condensers just follows by replacing $[\cond]$ with $[\cond]_{S}$ in Theorem~\ref{thm:cond-to-norm-quantum}. Note that even though an extractor is just a condenser with full output entropy $k'=m$, our norm characterization for condensers can not directly be used to characterize extractors because there is a loss of $\log(1/\eps)$ for the output entropy $k'$ (see Proposition~\ref{prop:cond-to-norm} and Theorem~\ref{thm:cond-to-norm-quantum}). However, for that reason we have the separate norm characterization for extractors in Section~\ref{sec:extr_results}.


\subsubsection*{Applications}

It is known that condensers are closely related to graph-theoretic problems~\cite{Umans07}, and here we make exactly such a connection (that is different from previously studied connections). Using the bounded norm characterization in Proposition~\ref{prop:cond-to-norm}, we can show that evaluating the performance of a condenser corresponds to an instance of a well studied combinatorial problem, called bipartite densest subgraph. For this we think of condensers as bipartite graphs $G=(N,M,V\subset N \times M,D)$ with left index set $N$, right index set $M$, edge set $V$ having left degree $D$, and neighbor function $\Gamma: N \times D \to M$. (Note the slight abuse in notation that we use throughout the paper: $N$, $M$ and $D$ refer both to sets and to their sizes.) The identification with the usual definition of condensers $\cond=\{f_{s}\}_{s\in D}$ is just by setting
\begin{align}\label{eq:neighbor_graph}
\Gamma(\cdot,s):=f_{s}(\cdot)
\end{align}
for every value of the seed $s\in D$.

\begin{restatable}{proposition}{bipartitegraph}\label{prop:bipartite_graph}
Let $\cond=\{f_{s}\}_{s\in D}$, $[\cond]$ as defined in~\eqref{eq:cond_map}, $G=(N,M,V,D)$ be the bipartite graph as defined in~\eqref{eq:neighbor_graph}, and $\mathrm{Dense}(G,2^{k},2^{k'})$ be the optimal value of the quadratic program
\begin{center}
\begin{minipage}[t]{0.90\textwidth}
\begin{align}
\mathrm{Dense}(G,2^{k},2^{k'})\eqdef \text{maximize}\quad & \sum_{(x,y) \in V}  f_x \, g_y \label{eq:conddensestsubgraph}\\
\text{subject to}\quad & f_x, g_y \in [0,1]\\
& \sum_x f_x \leq 2^{k}\\
& \sum_y g_y \leq 2^{k'}\ .
\end{align}
\end{minipage}
\end{center}
Then, we have
\begin{align}
\left\|[\cond]:\cap\{2^{k}\ell^\infty_N,\ell^1_N\}\to\Sigma\{2^{k'}\ell^\infty_M,\ell^1_M\}\right\|\leq\eps\quad\Leftrightarrow\quad\mathrm{Dense}(G,2^{k},2^{k'})\leq2^{k} D \eps\ .
\end{align}
\end{restatable}

In fact, it is well-known that the optimal value of the quadratic program will be achieved on $f_x, g_y \in \{0,1\}$. Hence, we try to find the subgraph of $G$ with the biggest number of edges, having $2^{k}$ vertices on the left and $2^{k'}$ vertices on the right. The norm condition for being a condenser (Proposition~\ref{prop:cond-to-norm}) then just says that the size of the edge set of all such subgraphs has to be bounded by $2^{k}D\eps$. This is exactly an instance of the bipartite densest $(2^{k},2^{k'})$-subgraph problem. The (bipartite) densest subgraph problem is well studied in theoretical computer science, but its hardness remains elusive. However, it is known that usual semidefinite program (SDP) relaxation possess a large integrality gap for random graphs~\cite{FS97,Bhaskara:2012wn}. Interestingly, we can show that the densest subgraph SDP relaxations are not only an upper bound on the quadratic program~\eqref{eq:conddensestsubgraph} characterizing the bounded norm of condensers, but also on the completely bounded norm of condensers.

Seeing condensers as bipartite graphs also allows us to define a Bell inequality (two-player game) such that the classical value characterizes the condenser property and the entangled value characterizes the quantum-proof condenser property (see~\cite{Brunner14} for a review article about Bell inequalities / two-player games). Starting from the bipartite graph $G=(N,M,V,D)$ as defined in~\eqref{eq:neighbor_graph}, we use operator space techniques by Junge~\cite{Junge:2006wt} to define a two-player game $(G;2^{k},2^{k'})$ with classical value $\omega(G;2^{k},2^{k'})$ and entangled value $\omega^{*}(G;2^{k},2^{k'})$, that has the following properties.

\begin{restatable}{theorem}{condensergame}\label{thm:condensergame}
Let $\cond=\{f_{s}\}_{s\in D}$, $[\cond]$ as defined in~\eqref{eq:cond_map}, $G=(N,M,V,D)$ be the bipartite graph as defined in~\eqref{eq:neighbor_graph}, and $(G;2^{k},2^{k'})$ be the two-player game defined in Section~\ref{sec:two_player}. Then, there is a constant $c>0$ such that
\begin{align}\label{eq:condensergame_b}
\omega(G;2^{k},2^{k'})\leq2^{-k'}\cdot\left\|[\cond]:\cap\left\{2^{k}\ell^\infty_N,\ell^1_N\right\}\to\Sigma\left\{2^{k'}\ell^\infty_M,\ell^1_M\right\}\right\|\leq c\cdot\omega(G;2^{k},2^{k'})\ ,
\end{align}
as well as
\begin{align}\label{eq:condensergame_cb}
\omega^{*}(G;2^{k},2^{k'})\leq2^{-k'}\cdot\left\|[\cond]:\capcap\left\{2^{k}\mathds{L}_N^\infty,\mathds{L}_N^1\right\}\to\sigsig\left\{2^{k'}\mathds{L}_M^\infty,\mathds{L}^1_M\right\}\right\|_\cb\leq c\cdot\omega^{*}(G;2^{k},2^{k'})\ .
\end{align}
Furthermore, the amount of entanglement used by the players corresponds to the dimension of the quantum side information $Q$.
\end{restatable}

Theorem~\ref{thm:condensergame} shows that for every condenser construction there is a corresponding Bell inequality, and that the degree to which this inequality is violated by quantum mechanics characterizes how quantum-proof the condenser construction is (and vice versa). So in particular, fully quantum-proof condensers have a corresponding Bell inequality that is not violated by quantum mechanics.


\section{Open problems}\label{sec:conclusions}

We showed how the theory of operator spaces provides a useful framework for studying the behavior of randomness extractors and condensers in the presence of quantum adversaries. However, there are many questions left open from here and we believe that the following are particularly interesting to explore:
\begin{itemize}
\item The main question if there exist good classical extractors that are not quantum-proof with a large gap still remains open. More precisely we would like to understand the classical/quantum separation better by finding tighter upper bounds (as in Theorem~\ref{thm:small-error-ent} and Theorem~\ref{thm:high-min-ent-ext}) as well as tighter lower bounds (as in~\cite{Gavinsky07}) on the size of the gap.
\item Is it possible to give an upper bound on the dimension of the quantum adversary that is sufficient to consider? This is also a natural question in the norm formulation (Theorem~\ref{thm:ext-to-norm-quantum} and Theorem~\ref{thm:cond-to-norm-quantum}) and the Bell inequality formulation (Theorem~\ref{thm:condensergame}). In the first case it translates into the question for what dimension the completely bounded norm saturates, and in the latter case it translates into the question how much entanglement is needed to achieve the optimal entangled value.
\item Given our new connection to Bell inequalities (Theorem~\ref{cor:out-size}) it would be interesting to explore the corresponding Bell inequalities of quantum-proof extractor constructions (since they can not be violated by quantum mechanics).
\item What other explicit extractor constructions are quantum-proof? This includes variations of Trevisan's constructions as, e.g., listed in~\cite[Section 6]{De12}, but also condenser based constructions~\cite{RR99,RSW00}. Here the motivation is that all of these constructions have better parameters than any known quantum-proof construction.
\item Operator space techniques might also be useful for analyzing fully quantum and quantum-to-classical randomness extractors as described in~\cite{Berta11_5,Berta12_2,Berta14}.
\end{itemize}


\section{Preliminaries}\label{sec:preliminaries}

Here we present some basic facts about quantum information theory (Section~\ref{sec:quantum_information}), normed spaces (Section~\ref{sec:spaces}), and operator spaces (Section~\ref{sec:operator_spaces}). A reader already familiar with these topics might consult Section~\ref{sec:notation} where we briefly subsume our notation.


\subsection{Quantum information}\label{sec:quantum_information}

In quantum theory, a system is described by an inner-product space, that we denote here by letters like $N, M, Q$.\footnote{In the following all spaces are assumed to be finite-dimensional.} Note that we use the same symbol $Q$ to label the system, the corresponding inner-product space and also the dimension of the space. Let $\matr_Q(S)$ be the vector space of $Q \times Q$ matrices with entries in $S$. Whenever $S$ is not specified, it is assumed to be the set of complex numbers $\bC$, i.e., we write $\matr_Q(\bC)=:\matr_{Q}$. The state of a system is defined by a positive semidefinite operator $\rho_{Q}$ with trace $1$ acting on $Q$. The set of states on system $Q$ is denoted by $\cS(Q) \subset \matr_Q$. The inner-product space of a composite system $QN$ is given by the tensor product of the inner-product spaces $Q \otimes N=:QN$. From a joint state $\rho_{QN} \in \cS(QN)$, we can obtain marginals on the system $Q$ by performing a partial trace of the $N$ system $\rho_Q \eqdef \tr_{N}[\rho_{QN}]$. A state $\rho_{QN}$ on $QN$ is called quantum-classical (with respect to some basis) if it can be written as $\rho_{QN} = \sum_{x} \rho_{x} \otimes \proj{x}$ for some basis $\{\ket{x}\}$ of $N$ and some positive semidefinite operators $\rho_x$ acting on $Q$. We denote the maximally mixed state on system $N$ by $\upsilon_N$.

To measure the distance between two states, we use the trace norm $\| A \|_{1} \eqdef \tr[\sqrt{A^{*}A}]$, where $A^{*}$ is the conjugate transpose of $A$. In the special case when $A$ is diagonal, $\| A \|_{1}$ becomes the familiar $\ell^1_N$-norm of the $N$ diagonal entries of $A$. Moreover, the Hilbert-Schmidt norm is defined as $\|A\|_{2}\eqdef\sqrt{\tr[A^{*}A]}$, and when $A$ is diagonal this becomes the usual $\ell^2_N$-norm. Another important norm we use is the operator norm, or the largest singular value of $A$, denoted by $\|A\|_{\infty}$. When $A$ is diagonal, this corresponds to the familiar $\ell^\infty_N$-norm. For a probability distribution $P_N \in \cS(N)$, $\| P_N \|_{\ell^\infty_N}$ corresponds to the optimal probability with which $P_N$ can be guessed. We write
\begin{align}
H_{\min}(N)_P\eqdef -\log \| P_N \|_{\ell^\infty_N}\quad\text{for the min-entropy of $P_{N}$.}
\end{align}
More generally, the conditional min-entropy of $N$ given $Q$ is used to quantify the uncertainty in the system $N$ given the system $Q$. The conditional min-entropy is defined as 
\begin{align}
H_{\min}(N|Q)_{\rho}\eqdef-\log\min_{\sigma_{Q}\in\cS(Q)}\left\|\left(\1_{N}\otimes\sigma_{Q}^{-1/2}\right)\rho_{NQ}\left(\1_{N}\otimes\sigma_{Q}^{-1/2}\right)\right\|_{(\infty;\infty)}\ ,
\end{align}
with generalized inverses, and $\|\cdot\|_{(\infty;\infty)}$ the operator norm on $\matr_{Q}\otimes\matr_{N}=:\matr_{QN}$. Note that in the special case where the system $Q$ is trivial, we have $H_{\min}(N)_{\rho}=-\log\|\rho_{N}\|_{\infty}$. In fact, the general case also corresponds to a norm, we have
\begin{align}
H_{\min}(N|Q)_{\rho}=-\log\|\rho_{QN}\|_{(1;\infty)}\ ,
\end{align}
where the norm $\| \cdot \|_{(1;\infty)}$ on $\matr_{QN}$ is defined as
\begin{align}\label{eq:def-one-infty}
\| A \|_{(1;\infty)} \eqdef \inf \Big\{ \| B_1 \|_2 \| C \|_{(\infty;\infty)} \| B_2 \|_2 : A = (B_1 \otimes \1_{N}) C (B_2 \otimes \1_{N}); B_1, B_2 \in \matr_{Q}\Big\}\ .
\end{align}
A proof of this is given in the Appendix as Proposition~\ref{prop:min-ent-norm}.


\subsection{Normed spaces}\label{sec:spaces}

A vector space $E$ together with a norm $\|\cdot\|_{E}$ defines a normed space denoted by $\mathbf{E}=\left(E,\|\cdot\|_{E}\right)$. On the dual vector space $E^{*}$ the dual norm $\|\cdot\|_{E^{*}}$ is defined as
\begin{align}
\|f\|_{E^{*}}\eqdef\sup_{\|e\|_{E}\leq1}|f(e)|\ ,
\end{align}
and hence $\mathbf{E}^*:=(E^{*},\|\cdot\|_{E^{*}})$ is again a normed space. We have that $\mathbf{E}$ is isomorphic to $\mathbf{E}^{**}$ since we restrict to finite-dimensional spaces.

The vector space of linear operators from $E$ to $F$ is denoted by $\lin(E,F)$ and for normed spaces $\mathbf{E}=\left(E,\|\cdot\|_{E}\right)$ and $\mathbf{F}=\left(F,\|\cdot\|_{F}\right)$ there is a natural norm induced on $\lin(E,F)$ defined by
\begin{align}\label{eq:bounded}
\|u\|_{\mathbf{E}\to\mathbf{F}}\eqdef\sup_{\|x\|_E \leq 1} \|u(x)\|_F\ ,
\end{align}
where $u\in\lin(E,F)$. This norm is called the bounded norm and we also use the notation
\begin{align}
\|u:\mathbf{E}\to\mathbf{F}\|\eqdef \|u\|_{\mathbf{E}\to\mathbf{F}}\ .
\end{align}
We write $\mathbf{B}(\mathbf{E},\mathbf{F})=\left(\lin(E,F),\|\cdot\|_{\mathbf{E}\to\mathbf{F}}\right)$ for the resulting normed space. Note that every $u\in\lin(E,F)$ has finite bounded norm (since we restrict to finite-dimensional spaces). The dual of the bounded norm $\mathbf{E}\to\mathbf{F}$ is the bounded norm $\mathbf{F^*}\to\mathbf{E^*}$. We define the adjoint map $u^*\in\lin(F^*,E^*)$ as $(u^*(f^*))(e) = f^*(u(e))$ for $e \in E$ and $f^* \in F^*$. It is simple to check that we have
\begin{align}\label{eq:dual-norm-map}
\|u\|_{\mathbf{E}\to\mathbf{F}} = \|u^*\|_{\mathbf{F}^*\to\mathbf{E}^*}\ .
\end{align}

The norms we are interested in are constructed by combining the $\ell^\infty_N$-norm and the $\ell^1_N$-norm. The following gives a general way of combining two norms.

\begin{definition}[$\cap$-norm]
Let $\mathbf{E}^{\alpha}=\left(E,\|\cdot\|_{\alpha}\right)$ and $\mathbf{E}^{\beta}=\left(E,\|\cdot\|_{\beta}\right)$ be two normed spaces on the same vector space $E$. We define a new normed space $\mathbf{E}^{\cap\{\alpha,\beta\}}=\left(E,\|\cdot\|_{\cap\{\alpha,\beta\}}\right)$ by
\begin{align}
\|\cdot\|_{\cap\{\alpha,\beta\}}\eqdef\max\big\{\|\cdot\|_{\alpha}, \|\cdot\|_{\beta}\big\}\ .
\end{align}
The dual of the normed space $\mathbf{E}^{\cap\{\alpha,\beta\}}$ is denoted $\mathbf{E}^{\Sigma\{\alpha^{*},\beta^{*}\}}=\left(E^*,\|\cdot\|_{\Sigma\{\alpha^{*},\beta^{*}\}}\right)$.
\end{definition}

The dual norm then takes the following form.

\begin{proposition}\label{prop:cap-sigma}
The normed space $\mathbf{E}^{\Sigma\{\alpha^{*},\beta^{*}\}}=\left(E^*,\|\cdot\|_{\Sigma\{\alpha^{*},\beta^{*}\}}\right)$ is given by
\begin{align}
\text{$\| y \|_{\Sigma\{\alpha^{*},\beta^{*}\}} = \inf\big\{\| y_1 \|_{\alpha^*} + \| y_2 \|_{\beta^*} : y_1 + y_2 = y \big\}$ with $y\in E^{*}$.}
\end{align}
\end{proposition}

\begin{proof}
For a given $y \in E^*$, we write the dual convex program to $\| y \|_{\Sigma\{\alpha^{*},\beta^{*}\}} = \inf \{ \| y_1 \|_{\alpha^*} + \| y_2 \|_{\beta^*} : y = y_1 + y_2 \}$ as
\begin{align}
\sup_{x \in E} \inf_{y_1 y_2} \| y_1 \|_{\alpha^*} + \| y_2 \|_{\beta^*} + (y - y_1 - y_2)(x)\ .
\end{align}
But observe that $\inf_{y_1} \| y_1 \|_{\alpha^*} - y_1(x) = 0$ if $\| x \|_{\alpha} = \sup_{y_1} \{ y_1(x) :  \| y_1 \|_{\alpha^*} \leq 1 \} \leq 1$ and $\inf_{y_1} \| y_1 \|_{\alpha^*} - y_1(x) = - \infty$ otherwise. Thus the dual convex program can be written as
\begin{align}
\sup_{x \in E} \{ y(x) : \|x\|_{\alpha} \leq 1, \| x \|_{\beta} \leq 1\}\ ,
\end{align}
which is the definition of the dual of the $\cap\{\alpha,\beta\}$-norm. We conclude by observing that strong duality holds in this case. 
\end{proof}


\subsection{Operator spaces}\label{sec:operator_spaces}

For the convenience of the reader we recall a few basic facts about the theory of operator spaces, sometimes also referred to as ``quantized Banach spaces''. For a more in depth treatment we refer the reader to Pisier's book~\cite{Pis03}.

\begin{definition}[Operator space]\label{def:operator_space}
An (abstract) operator space $\mathds{E}$ is a vector space $E$ together with a sequence of normed spaces $\mathbf{E}_{Q}=\left(\matr_Q(E),\|\cdot\|_{\mathds{E}_Q}\right)$ such that for all $Q\geq1$:
\begin{enum}
\item For all $x \in \matr_Q(E)$ and $x' \in \matr_Q(E)$\ ,
\begin{align}
\left\| 
\left( \begin{array}{cc}
x & 0 \\
0 & x' \\
\end{array} 
\right)
\right\|_{\mathds{E}_{Q+Q'}} = \max\left\{ \|x\|_{\mathds{E}_Q}, \|x'\|_{\mathds{E}_{Q'}} \right\}\ .\label{eq:cond-direct-sum}
\end{align}
\item For all $\alpha, \beta \in \matr_Q$ and $x \in \matr_Q(E)$\ ,
\begin{align}\label{eq:cond-mult}
\| \alpha \cdot x \cdot \beta \|_{\mathds{E}_Q} \leq \| \alpha \|_{\infty} \| x \|_{\mathds{E}_Q} \| \beta \|_{\infty}\ ,
\end{align}
where the notation $\alpha \cdot x$ refers to usual multiplication of matrices.
\end{enum}
We write $\mathds{E}=\left(E,\|\cdot\|_{\mathds{E}_Q}\right)$ for the operator space structure. Moreover, we abbreviate $\mathbf{E}\equiv\mathbf{E}_1$ with $\|\cdot\|_E\equiv\|\cdot\|_{\mathds{E}_1}$ for the normed space with $Q=1$. We say that $\mathds{E}$ is an operator space on the vector space $E$, or on the normed space $\mathbf{E}$ if we want to specify the $\mathds{E}_1$-norm.
\end{definition}

The most important example of an operator space is
\begin{align}\label{eq:MN_operator}
\text{$\mathds{L}^{\infty}_N:=\left(\matr_N,\|\cdot\|_{\left(\mathds{L}^{\infty}_N\right)_Q}\right)$ with $\|\cdot\|_{\left(\mathds{L}^{\infty}_N\right)_Q}:=\|\cdot\|_{(\infty;\infty)}$ the usual operator norms on $\matr_{QN}$.}
\end{align}
It is easy to verify that the two conditions~\eqref{eq:cond-direct-sum} and~\eqref{eq:cond-mult} are satisfied. We note that this operator space is usually called $M_N$.

Alternatively we could also define a (concrete) operator space $\mathds{E}=\left(E,\|\cdot\|_{\mathds{E}_Q}\right)$ on a normed space $\mathbf{E}=(E,\|\cdot\|_{E})$ with $\|\cdot\|_E\equiv\|\cdot\|_{\mathds{E}_1}$ by seeing $\mathbf{E}$ as a subspace of a normed space
\begin{align}\label{eq:MN_normed}
\mathbf{L}^{\infty}_N:=(\matr_N,\|\cdot\|_{\infty}),\quad\text{with $\|\cdot\|_{\mathds{E}_Q}$ induced by the operator norms $\|\cdot\|_{(\infty;\infty)}$ on $\matr_{QN}$.}
\end{align}
Since Ruan~\cite{Rua88} proved that every abstract operator space can be realized as a concrete operator space (see also~\cite[Chapter 2.2]{Pis03}) these two definitions are really equivalent.\footnote{In general, one needs to embed $\mathbf{E}=(E,\|\cdot\|_{E})$ into the space of bounded operators on an infinite-dimensional Hilbert space.} For example, consider the subspace $D_N\subset\matr_N$ of diagonal matrices. It is then immediate to deduce that for this concrete operator space $\mathds{D}_N=\left(D_N,\|\cdot\|_{(\mathds{D}_N)_Q}\right)$, $\|\cdot\|_{D_N}$ is simply the $\ell^\infty_N$-norm of the diagonal vector. As another example, consider the subspace $C_N$ ($R_{N}$) of matrices such that only the first column (row) contains non-zero elements. This defines concrete operator spaces $\mathds{C}_N=\left(C_N,\|\cdot\|_{(\mathds{C}_N)_Q}\right)$ and $\mathds{R}_N=\left(R_N,\|\cdot\|_{(\mathds{R}_N)_Q}\right)$, respectively, with $\|\cdot\|_{C_N}=\|\cdot\|_{R_N}$ simply given by the $\ell^2_N$-norm of the vector. Note that even though the normed spaces $\mathbf{C}_{N}=\left(C_N,\|\cdot\|_{C_N}\right)$ and $\mathbf{R}_{N}=\left(R_N,\|\cdot\|_{R_N}\right)$ are both isomorphic to $\left(\bC^N,\|\cdot\|_{\ell^2_N}\right)$, the operator spaces $\mathds{C}_N$ and $\mathds{R}_N$ are different. In fact, a given normed space has in general many possible operator space structures. However, for many normed spaces that we are interested in, there is one ``natural'' operator space structure. For example, for the normed space $\mathbf{L}^\infty_N$ as in~\eqref{eq:MN_normed} there is the natural operator space structure $\mathds{L}_N^{\infty}$ as in~\eqref{eq:MN_operator}. Also, for the normed space
$\mathbf{L}_N^1:=(\matr_N,\|\cdot\|_{1})$ there is a natural operator space structure $\mathds{L}_N^1$, defined as the operator space dual of $\mathds{L}_N^{\infty}$. We will discuss this in detail after the definition of the operator space dual (Definition~\ref{prop:s1-os}).

The bounded norm as in~\eqref{eq:bounded} is fundamental in understanding linear maps between normed spaces. The analogous norm for linear maps between operator spaces is the completely bounded norm (because completely bounded maps are the morphisms in the category of operator spaces). In the quantum information literature, the completely bounded norm usually refers specifically to maps between the operator spaces $\mathds{L}_N^\infty$ and $\mathds{L}_M^\infty$. The dual norm is called the diamond norm, first introduced in quantum information by Kitaev~\cite{Kit97}. Here we are concerned with the completely bounded norm between more general operator spaces.

\begin{definition}[Completely bounded norm]
For operator spaces $\mathds{E}=\left(E,\|\cdot\|_{\mathds{E}_Q}\right)$ and $\mathds{F}=\left(F,\|\cdot\|_{\mathds{F}_Q}\right)$ the completely bounded norm of $u\in\lin(E,F)$ is defined as
\begin{align}\label{eq:cb-norm}
\|u\|_{\mathds{E}\to\mathds{F}}\eqdef\sup_Q\|u_Q\|_{\mathbf{E}_Q\to\mathbf{F}_Q}\ ,
\end{align}
where for $\{x_{ij}\}_{ij} \in \matr_Q(E)$, $u_Q(\{x_{ij}\}_{ij})\eqdef\{u(x_{ij})\}_{ij}$, or simply $u_Q = \1_Q \otimes u$. We also use the notation
\begin{align}
\|u:\mathds{E}\to\mathds{F}\|_{\cb}\eqdef \|u\|_{\mathds{E}\to\mathds{F}}\ .
\end{align}
\end{definition}

Hence $\left(\lin(E,F),\|\cdot\|_{\mathds{E}\to\mathds{F}}\right)$ is a normed space just as $\mathbf{B}(\mathbf{E},\mathbf{F})=(\lin(E,F),\|\cdot\|_{\mathbf{E}\to\mathbf{F}})$ is.  Note however that even though every $u\in\lin(E,F)$ has finite completely bounded norm (since we restrict to finite-dimensional spaces), in general $\|\cdot\|_{\mathbf{E}\to\mathbf{F}}$ is smaller than $\|\cdot\|_{\mathds{E}\to\mathds{F}}$. Later, we are interested in upper bounding this ratio for particular operator spaces and maps. For general $\mathds{E},\mathds{F}$ and $u\in\lin(E,F)$ of rank $M$ it is known that the ratio of the completely bounded to the bounded norm is at most $M/2^{1/4}$~\cite[Theorem 7.15]{Pis03}.

It is in general not true that we can restrict the supremum in~\eqref{eq:cb-norm} to finite $Q$. However, for the target operator space $\mathds{L}_N^{\infty}$ defined in~\eqref{eq:MN_operator}, we have~\cite[Proposition 1.12]{Pis03}:
\begin{align}\label{eq:cb-norm-to-mn}
\|u\|_{\mathds{E}\to\mathds{L}_N^{\infty}}=\|u_N\|_{\mathbf{E}_N\to(\mathbf{L}_N^{\infty})_{N}}\ .
\end{align}
This raises the question whether there are specific operator space structures such that all bounded norms are also completely bounded. Such structures do in fact exist, and are called minimal and maximal operator space structures.

\begin{proposition}\cite[Chapter 3]{Pis03}
Let $\mathbf{E}=(E,\|\cdot\|_{E})$ be a normed space. Then there exists two (in general different) operator spaces $\mathds{E}^{\max}$ and $\mathds{E}^{\min}$ on $\mathbf{E}$, such that we have for all operator spaces $\mathds{F}=\left(F,\|\cdot\|_{\mathds{F}_Q}\right)$ with $u\in\lin(E,F)$ and $v\in\lin(F,E)$:
\begin{align}
\norm{u}_{\mathds{F}\to\mathds{E}^{\min}} = \norm{u}_{\mathbf{F}\to\mathbf{E}} \text{ and }\norm{v}_{\mathds{E}^{\max}\to\mathds{F}} = \norm{v}_{\mathbf{E}\to \mathbf{F}}\ .
\end{align}
\end{proposition}

For $\mathds{E}=\left(E,\|\cdot\|_{\mathds{E}_Q}\right)$ and $\mathds{F}=\left(F,\|\cdot\|_{\mathds{F}_Q}\right)$ operator spaces there is a natural operator space structure $\mathds{B}(\mathds{E},\mathds{F})=\left(\lin(E,F),\|\cdot\|_{\mathds{B}(\mathds{E},\mathds{F})_Q}\right)$. For the definition of the $\mathds{B}(\mathds{E},\mathds{F})_Q$-norms, observe that we can see $\matr_Q(\lin(E,F))$ as elements of $\lin(E,\matr_Q(F))$. Thus for $x \in \matr_Q(\lin(E,F))$, we define
\begin{align}
\| x \|_{\mathds{B}(\mathds{E},\mathds{F})_Q}\eqdef \| x \|_{\mathds{E}\to\mathds{L}^{\infty}_Q(\mathds{F})}\quad\text{with the operator space $\mathds{L}^{\infty}_Q(\mathds{F}):=\left(\matr_Q(F),\|\cdot\|_{(\mathds{L}^{\infty}_Q(\mathds{F}))_{Q'}}\right)$}
\end{align}
having $\|\cdot\|_{(\mathds{L}^{\infty}_Q(\mathds{F}))_{Q'}}:=\|\cdot\|_{\mathds{F}_{QQ'}}$. It is then simple to verify the two conditions~\eqref{eq:cond-direct-sum} and~\eqref{eq:cond-mult}. By taking $F=\bC$, this allows us to define the notion of a dual operator space.

\begin{definition}[Operator space dual]\label{def:opertordual}
For an operator space $\mathds{E}=\left(E,\|\cdot\|_{\mathds{E}_Q}\right)$, the dual operator space $\mathds{E}^*=\left(E^*,\|\cdot\|_{\mathds{E}^*_Q}\right)$ is defined as
\begin{align}
\| x \|_{\mathds{E}^*_Q}\eqdef \| x \|_{\mathds{E}\to\mathds{L}^{\infty}_Q}\quad\text{with $\mathds{L}^{\infty}_Q$ as in~\eqref{eq:MN_operator} and $x \in \matr_Q(E^*)$ is viewed as an element of $\lin(E,\matr_Q)$.}
\end{align}
\end{definition}

For $u\in\lin(E,F)$ we have for the adjoint $u^*\in\lin(F^{*},E^{*})$ that $\| u \|_{\mathds{E}\to\mathds{F}} = \| u^* \|_{\mathds{F}^*\to\mathds{E}^*}$. We also have $\mathds{E}^{**}=\mathds{E}$ with $\mathbf{E}^{**}_Q=\mathbf{E}_Q$ since we restrict to finite-dimensional spaces. If we consider the norm for $Q = 1$, then $\|\cdot\|_{\mathds{E}^{*}_1}$ corresponds to the dual norm of $\|\cdot\|_{\mathds{E}_1}$, since $\|\cdot\|_{\mathds{E}\to\mathds{L}_1^\infty}=\|\cdot\|_{\mathbf{E}\to\mathbf{L}_1^\infty}$~\cite[Proposition 1.10]{Pis03}. However, this is not the case for $Q\geq 2$, that is, $\|\cdot\|_{\mathds{E}^*_Q}$ on $\matr_{Q}(E^{*})$ is in general not the dual norm of $\|\cdot\|_{\mathds{E}_Q}$ on $\matr_{Q}(E)$.

As an example let us now consider the dual operator space of $\mathds{L}_N^{\infty}$, called the trace class operator space $\mathds{L}_N^1:=\left(\mathds{L}_N^{\infty}\right)^*$.

\begin{proposition}\label{prop:s1-os}
For the trace class operator space $\mathds{L}_N^1=\left(\matr_N^*,\|\cdot\|_{(\mathds{L}_N^1)_Q}\right)$ we have
\begin{align}\label{eq:l1-os}
\|\cdot\|_{(\mathds{L}_N^1)_Q}=\|x\|_{(\infty;1)}\eqdef\sup\big\{\|(A\otimes\1_{N})x(B\otimes\1_{N})\|_{(1;1)}:\|A\|_2,\|B\|_2 \leq1;A,B\in\matr_Q\big\}\ ,
\end{align}
where $x \in \matr_Q(\matr_N^*)$ (with $\matr_N^*\equiv\matr_N$ as vector spaces), and $\|\cdot\|_{(1;1)}$ denotes the trace norm on $\matr_{QN}$.
\end{proposition}

We observe that the $(\infty;1)$-norm is the dual norm of the $(1;\infty)$-norm characterizing the conditional min-entropy~\eqref{eq:def-one-infty}; see Proposition~\ref{prop:duality-one-infty}.

\begin{proof}
In order to compute the dual norm, we first explicitly map $x \in \matr_Q(\matr_N)$ to an element of $\lin(\matr_N,\matr_Q)$. For this, we see $x$ as the Choi matrix of a map $u \in \lin(\matr_N,\matr_Q)$ defined by
\begin{align}
\tr\left[du(c)\right]=\tr\left[x\left(d\otimes c^T\right)\right]
\end{align}
for any $c\in\matr_N,d\in\matr_Q$, where $c^T$ denotes the transpose of $c$ in the standard basis. Using the definition of operator space dual (Definition~\ref{def:opertordual}), we have that
\begin{align}
\|x\|_{(\mathds{L}_N^1)_Q}=\left\|u:\mathds{L}_N^\infty\to\mathds{L}_Q^\infty\right\|_{\cb}=\left\|u:\left(\mathds{L}_N^\infty\right)_Q\to\left(\mathds{L}_Q^\infty\right)_Q\right\|\ ,
\end{align}
where we used~\eqref{eq:cb-norm-to-mn} to simplify the completely bounded norm. Continuing we get,
\begin{align}
\|x\|_{(\mathds{L}_N^1)_Q}
&= \sup\left\{ \| (\1_Q \otimes u)(z) \|_{(\infty;\infty)} : z \in \matr_{QN}, \| z \|_{(\infty;\infty)} \leq 1\right\} \\
&=\sup\left\{\tr\left[\left(\1_Q\otimes u\right)(z)ab^{*}\right]:z\in\matr_{QN},\|z\|_{(\infty;\infty)}\leq1;a,b\in\C^{Q^2},\|a\|_2,\|b\|_2\leq1\right\}\ .
\end{align}
A straightforward calculation (Lemma~\ref{lem:calc-vect-to-matr}) shows that
\begin{align}
\tr\left[ (\1_Q \otimes u)(z) a b^{*} \right]= \tr\left[ (\bar{B} \otimes \1_{N}) x (A^T \otimes \1_{N}) z^{T} \right]\ ,
\end{align}
where $A\eqdef \sum_{ij} a_{ij} \ket{i} \bra{j}$ and $B\eqdef \sum_{ij} b_{ij} \ket{i} \bra{j}$. To conclude, we use that $\| z^T \|_{(\infty;\infty)} = \| z \|_{(\infty;\infty)}$, $\|A^T \|_2 = \| a \|_2$, $\| \bar{B} \|_2 = \| b \|_2$, and that the trace norm is the dual norm of the operator norm.
\end{proof}

We now can define the analogs of the $\cap$- and $\Sigma$-norms for operator spaces. For two sequences of norms $\|\cdot\|_{\alpha_Q}$ and $\|\cdot\|_{\beta_Q}$ defining operator space structures on the same vector space $E$, we denote by $\|\cdot\|_{\alpha^*_Q}$ and $\|\cdot\|_{\beta^*_Q}$ the sequence of norms on the dual space $E^*$ giving rise to the operator space duals.

\begin{definition}[$\cap$-norm]\label{def:cap-sigma-os} 
Let $\mathds{E}^{\alpha}=\left(E,\|\cdot\|_{\alpha_Q}\right)$ and $\mathds{E}^{\beta}=\left(E,\|\cdot\|_{\beta_Q}\right)$ be two operator spaces on the same vector space $E$. We define a new operator space $\mathds{E}^{\cap\{\alpha,\beta\}}=\left(E,\|\cdot\|_{(\cap\{\alpha,\beta\})_Q}\right)$ by the sequence of norms
\begin{align}
\norm{x}_{(\cap\{\alpha,\beta\})_Q}\eqdef \max\left\{\norm{x}_{\alpha_Q},\norm{x}_{\beta_Q}\right\}\ .
\end{align}
The operator space dual of $\mathds{E}^{\cap\{\alpha,\beta\}}$ is denoted $\mathds{E}^{\Sigma\{\alpha^*,\beta^*\}}=\left(E^*,\|\cdot\|_{(\Sigma\{\alpha^*,\beta^*\})_Q}\right)$.
\end{definition}

One might think that the $\left(\Sigma\{\alpha^*,\beta^*\}\right)_Q$-norms are equal to the $\Sigma$-norms of the $\alpha^*_Q$- and $\beta^*_Q$-norms. This is almost the case, but only up to a factor of $2$; see~\cite[Chapter 2.7]{Pis03}. A more detailed discussion of the $\cap$- and $\Sigma$-operator space structures can be found in~\cite[Chapter 2.7]{Pis03}.


\subsection{Notation}\label{sec:notation}

The purpose of this subsection is to summarize the important notation we introduced. The tensor product of inner-product spaces $Q$ and $N$ is denoted by $Q \otimes N=QN$. The vector space of $Q \times Q$ matrices with entries in $S$ is denoted by $\matr_Q(S)$. Whenever $S$ is not specified, it is assumed to be the set of complex numbers $\bC$, i.e., $\matr_Q(\bC)=\matr_{Q}$. We also use the abbreviation $\matr_{Q}\otimes\matr_{N}=\matr_{QN}$. The set of positive semidefinite operators with trace one acting on $Q$ is denoted by $\cS(Q) \subset \matr_{Q}$.

A vector space $E$ together with a norm $\|\cdot\|_{E}$ defines the normed space $\mathbf{E}=\left(E,\|\cdot\|_{E}\right)$. The vector space of linear operators from $E$ to $F$ is denoted by $\lin(E,F)$. For normed spaces $\mathbf{E}=(E,\|\cdot\|_E)$ and $\mathbf{F}=(F,\|\cdot\|_F)$ there is a norm induced on $\lin(E,F)$ defined by
\begin{align}
\|u\|_{\mathbf{E}\to\mathbf{F}}=\|u:\mathbf{E}\to\mathbf{F}\|=\sup_{\|x\|_E \leq 1} \|u(x)\|_F,\quad\text{leading to the normed space $\mathbf{B}(\mathbf{E},\mathbf{F})=\left(\lin(E,F),\|\cdot\|_{\mathbf{E}\to\mathbf{F}}\right)$.}
\end{align}
A vector space $E$ together with a sequence of norms $\|\cdot\|_{\mathds{E}_Q}$ on the spaces $\matr_Q(E)$ (satisfying some consistency conditions) defines an operator space denoted $\mathds{E}=\left(E,\|\cdot\|_{\mathds{E}_Q}\right)$. For operator spaces $\mathds{E}=\left(E,\|\cdot\|_{\mathds{E}_Q}\right)$ and $\mathds{F}=\left(F,\|\cdot\|_{\mathds{F}_Q}\right)$ the completely bounded norm of $u\in\lin(E,F)$ is defined as
\begin{align}
\|u\|_{\mathds{E}\to\mathds{F}}=\|u:\mathds{E}\to\mathds{F}\|_{\cb}=\sup_{Q}\|\1_Q \otimes u\|_{\mathbf{E}_Q\to \mathbf{F}_Q}\ .
\end{align}
For the operator space $\mathds{B}(\mathds{E},\mathds{F})=\left(\lin(E,F),\|\cdot\|_{\mathds{B}(\mathds{E},\mathds{F})_Q}\right)$ we see $\matr_Q(\lin(E,F))$ as elements of $\lin(E,\matr_Q(F))$, and thus for $x \in \matr_Q(\lin(E,F))$,
\begin{align}
\| x \|_{\mathds{B}(\mathds{E},\mathds{F})_Q}=\| x \|_{\mathds{E}\to\mathds{L}^{\infty}_Q(\mathds{F})}\quad\text{with the operator space $\mathds{L}^{\infty}_Q(\mathds{F})=\left(\matr_Q(F),\|\cdot\|_{(\mathds{L}^{\infty}_Q(\mathds{F}))_{Q'}}\right)$,}
\end{align}
having $\|\cdot\|_{(\mathds{L}^{\infty}_Q(\mathds{F}))_{Q'}}=\|\cdot\|_{\mathds{F}_{QQ'}}$.


\section{Extractors}\label{sec:extractors}

\subsection{Extractor property as a bounded norm}

Here we characterize extractors in terms of the bounded norm $\cap\{2^{k}\ell^\infty_N,\ell^1_N\}\to\ell^1_M$.

\extToNorm*

\begin{proof}
We first prove~\eqref{eq:extractor_1}. For any probability distribution $P$ with min-entropy at least $k$ we have $\| P\|_{\ell^1_N} = 1$ as well as $\| P\|_{2^{k} \ell^\infty_N} \leq 1$. Hence, $\| P\|_{\cap(2^{k} \ell^\infty_N, \ell^1_N)} \leq 1$ and by the definition of the bounded norm this implies the claim $\|\Delta[\ext](P)\|_{\ell^1_M}\leq\eps$.

To prove~\eqref{eq:extractor_2} consider a distribution $P$ with $\|P\|_{\ell^1_N}\leq1$ and $\| P \|_{\ell^\infty_N} \leq 2^{-k}$. Then let $P^+(x) = \max\{P(x), 0\}$ and $P^-(x) = \max\{-P(x), 0\}$, and note that $\| P^+ \|_{\ell^\infty_N}, \| P^{-} \|_{\ell^\infty_N} \leq 2^{-k}$. As the extractor property only applies for normalized distributions, we extend $P^+$, $P^-$ into a probability distributions $\bar{P}^+ = P^+ + (1-\|P^+\|_{\ell^1_N})\upsilon_N$ and $\bar{P}^-(x) = P^-(x) +(1-\|P^-\|_{\ell^1_N}) \upsilon_N$. Now observe that $\|\bar{P}^+\|_{\ell^\infty_N}\leq\| P\|_{\ell^\infty_N}+\frac{1}{N}\leq2^{-(k-1)}$ and the similar bound holds for $\|\bar{P}^-\|_{\ell^\infty_N}$. Thus, we have
\begin{align}
\| \Delta[\ext](P) \|_{\ell^1_M} 
&= \| \Delta[\ext](P^+) - \Delta[\ext](P^-) \|_{\ell^1_M} \\
&\leq \| \Delta[\ext](P^+) \|_{\ell_1}  +  \| \Delta[\ext](P^-) \|_{\ell^1_M} \\
&\leq \| \Delta[\ext](\bar{P}^+) \|_{\ell^1_M} + (1-\|P^+\|_{\ell^1_M}) \| \Delta[\ext](\upsilon_N) \|_{\ell^1_M}+\| \Delta[\ext](\bar{P}^-) \|_{\ell^1_M}\notag\\
&\quad+ (1-\|P^-\|_{\ell^1_M}) \| \Delta[\ext](\upsilon_N) \|_{\ell^1_M}\\
&\leq \| \Delta[\ext](\bar{P}^+) \|_{\ell^1_M} + \| \Delta[\ext](\bar{P}^-) \|_{\ell^1_M} + 2\| \Delta[\ext](\upsilon_N) \|_{\ell^1_M}\\
&\leq 4\eps\ ,
\end{align}
where we used the fact that $\bar{P}^+$, $\bar{P}^-$ and $\upsilon_N$ have min-entropy at least $k-1$.
\end{proof}


\subsection{Quantum-proof as a completely bounded norm}\label{sec:norms-ext}

Recall that the relevant norm for extractors is a maximum of two norms denoted $\cap\{2^k\ell^\infty_N,\ell^1_N\}$. Our objective is to extend this norm to matrices so that it gives an operator space and so that the requirement of quantum-proof extractors is captured. In Section~\ref{sec:operator_spaces} we discussed operator spaces for the $\ell^\infty_N$-norm as well as for the $\ell^1_N$-norm. Moreover, it is simple to define an operator space for the maximum of two norms (intersection norm), because the corresponding unit ball is the intersection of the two unit balls of the norms. However, it turns out that because of positivity issues, this norm is not the most adapted for our purpose.

\begin{definition}[Intersection norms for extractors]\label{def:intersection_norms}
We define the two operator spaces
\begin{align}
\cap\{K\mathds{L}_N^\infty,\mathds{L}_N^1\}=\left(\matr_N,\|\cdot\|_{(\cap\{K\mathds{L}_N^\infty,\mathds{L}_N^1\})_Q}\right)\quad\mathrm{and}\quad\capcap\{K\mathds{L}_N^\infty,\mathds{L}_N^1\}=\left(\matr_N,\|\cdot\|_{(\capcap\{K\mathds{L}_N^\infty,\mathds{L}_N^1\})_Q}\right)\ .
\end{align}
for $x\in\matr_Q(\matr_N)$ as
\begin{align}\label{eq:def-cap}
&\| x \|_{(\cap\{K\mathds{L}_N^\infty,\mathds{L}_N^1\})_Q} \eqdef \max\left\{ K\| x \|_{(\infty;\infty)}, \| x \|_{(\infty;1)}\right\}\\
&\label{eq:def-cap-cap}
\| x \|_{(\capcap\{K\mathds{L}_N^\infty,\mathds{L}_N^1\})_Q} \eqdef \inf \Big\{\max\left\{ \sqrt{K} \| a \|_{(\infty;\infty)}, \| \Gamma (a) \|_{\infty} \right\} \cdot \max\left\{ \sqrt{K} \| b \|_{(\infty;\infty)}, \| \Gamma(b) \|_{\infty} \right\} : x = a b^*\Big\} \ ,
\end{align}
where $\Gamma\in\lin\left(\matr_{QN},\matr_Q\left(\mathbb{C}^{N^2}\right)\right)$ is defined as 
\begin{align}\label{eq:def-gamma}
\Gamma( \ket{i} \bra{j} \otimes \ket{k} \bra{\ell} ) = \ket{i} \bra{j} \otimes  \bra{k} \otimes \bra{\ell} \ .
\end{align}
\end{definition}

We also use the abbreviation
\begin{align}
\cap=\cap\{K\mathds{L}_N^\infty,\mathds{L}_N^1\}\quad\mathrm{and}\quad\cap\cdot\;\cap=\capcap\{K\mathds{L}_N^\infty,\mathds{L}_N^1\}\ .
\end{align}
It might appear that the $\capcap$-norm is rather artificial but it can in fact be constructed from basic operator spaces (the row and column operators spaces) and a fundamental operator space tensor product called the Haagerup tensor product. For details on this we refer to Appendix~\ref{app:intersection}.

For extractors we are naturally interested in the case of matrices which are diagonal with respect to the first system, that is, elements of $\matr_Q(\mathbb{C}^N)\subset\matr_Q(\matr_N)$. The $\cap$-norm and the $\capcap$-norm are then defined on $\matr_Q(\mathbb{C}^N)$ via the embedding of $\matr_Q(\mathbb{C}^N)$ into $\matr_{QN}$ as block-diagonal matrices.

The $\cap$-norm and the $\capcap$-norm are actually closely related and this can be seen from the following lemma, where we write the $(\infty;\infty)$-norm and $(\infty;1)$-norm as an infimum over all possible factorizations.

\begin{lemma}\label{lem:factorization-norms}
For $x \in \matr_{QN}$ we have
\begin{align}\label{eq:factorization-infty}
\| x \|_{(\infty; \infty)} = \inf \big\{ \| a \|_{(\infty; \infty)} \| b \|_{(\infty; \infty)} : x = a b^* \text{ and } a, b \in \matr_{QN} \big\}\ .
\end{align}
as well as
\begin{align}\label{eq:factorization-infty-one}
\| x \|_{(\infty; 1)} = \inf \big\{ \| \Gamma (a) \|_{\infty} \| \Gamma(b) \|_{\infty} : x = a b^* \text{ and } a, b \in \matr_{QN} \big\}\ ,
\end{align}
where $\Gamma\in\lin\left(\matr_{QN},\matr_Q\left(\mathbb{C}^{N^2}\right)\right)$ as defined in~\eqref{eq:def-gamma}.
\end{lemma}

Note that in contrast, the factorization in the $\capcap$-norm in~\eqref{eq:factorization-infty-one} is restricted to be the same one for both of the norms.

\begin{proof}
For one direction in~\eqref{eq:factorization-infty}, we have for any $x = a b^*$ that
\begin{align}
\| x \|_{(\infty;\infty)} \leq \| a \|_{(\infty;\infty)} \| b^* \|_{(\infty;\infty)}\ .
\end{align}
For the other direction, we write a polar decomposition of $x = UP$ with $U$ unitary and $P$ positive semidefinite. Then let $a = U \sqrt{P}$, $b = \sqrt{P}$, and hence $\| a \|_{(\infty;\infty)} = \| b \|_{(\infty;\infty)} = \| P \|_{(\infty;\infty)}^{1/2}$ which gives $\| a \|_{(\infty;\infty)} \| b \|_{(\infty;\infty)} = \| x \|_{(\infty;\infty)}$.

For~\eqref{eq:factorization-infty-one}, we use the definition of the $(\infty;1)$-norm in terms of the operator space dual of $\mathds{L}_N^\infty$. For that we see $x$ as the Choi matrix of $u\in\lin(\matr_N,\matr_{Q})$, defined by 
\begin{align}\label{eq:choi-matrix}
\tr\left[d u(c)\right] = \tr\left[ x \left(d \otimes c^T\right)\right]\ .
\end{align}
We then have $\| x \|_{(\infty; 1)} = \| u \|_{\mathds{L}^\infty_{N}\to\mathds{L}^\infty_{N}}$. Next we show the following useful claim:
\begin{align}
\text{$x = ab^*$ is equivalent to $u(z) = \hat{a} (z \otimes \1_{NQ}) \hat{b}^{*}$ for all $z \in \matr_N$, where $\hat{a} = \Gamma(a)$ and $\hat{b} = \Gamma(b)$.}
\end{align}
For this, let us write $x = \sum_{ii'kk'} x_{ii'kk'} \ket{i} \bra{i'} \otimes \ket{k} \bra{k'}$. Then $x = a b^*$ translates to $x_{ii'pq} = \sum_{j \ell} a_{ijp\ell} b^*_{i' j q \ell}$ for all $ii'pq$. On the other hand, we have that $\tr[ \ket{i'} \bra{i} u(\ket{p} \bra{q})] = \tr[ x  \ket{i'} \bra{i} \otimes \ket{q} \bra{p} ] = x_{ii'pq}$. We explicitly write
\begin{align}
\hat{a} = \sum_{ijk\ell} a_{ijk\ell} \ket{i} \bra{j} \otimes \bra{k} \otimes \bra{\ell} \quad\mathrm{and}\quad\hat{b}^* = \sum_{ijk\ell} b^*_{ijk\ell} \ket{j}\bra{i} \otimes \ket{k} \otimes \ket{\ell}\ .
\end{align}
As a result,
\begin{align}
\hat{a} \ket{p} \bra{q} \otimes \1_{NQ} \hat{b}^* &= \sum_{ii' j \ell } a_{ijp\ell} b^*_{i'jq\ell} \ket{i} \bra{i'}\ ,
\end{align}
which proves the claim. To finish the proof for~\eqref{eq:factorization-infty-one}, we use~\cite[Theorem 5]{Wat09} which states that\footnote{We refer to Watrous' paper as the notation and proof are quantum information friendly, but such a result is well known in operator space theory and holds in more generality, see e.g., \cite[Theorem 1.6]{Pis03}.}
\begin{align}
\| u \|_{\mathds{L}^\infty_{N}\to\mathds{L}^\infty_{N}}= \inf\left\{ \| \hat{a} \|_{\infty} \| \hat{b} \|_{\infty} : u(\beta) = \hat{a} \beta \otimes \1_{NQ} \hat{b}^*\right\}\ .
\end{align}
In~\cite{Wat09}, things are written in the dual form: the diamond norm of $u^*$ is considered and $\hat{a}^*$ and $\hat{b}^*$ is a Stinespring pair for the channel $u^*$ in the sense that $u^*(\beta) = \tr_{NQ}[\hat{a}^* \beta \hat{b}]$. Another point is that we can assume that the output dimension of $\hat{a}$ and $\hat{b}$ are $NQ^2$. This proves equality \eqref{eq:factorization-infty-one}.
\end{proof}

We now provide a simple bound on the $\capcap$-norm.

\begin{proposition}\label{prop:new-old-norm}
For $x \in \matr_Q (\matr_N)$ we have
\begin{align}\label{eq:new-old-norm1}
\text{$\| x \|_{(\capcap)_Q} \geq \| x \|_{\cap_Q}$ and if $x \geq 0$, $\| x \|_{(\capcap)_Q} = \| x \|_{\cap_Q}$.}
\end{align}
Moreover, for $Q = 1$, i.e., $x \in \matr_N$, we have $\| x \|_{\capcap} = \| x \|_{\cap}$. 
\end{proposition}

\begin{proof}
The inequality~\eqref{eq:new-old-norm1} follows directly from Definition~\ref{def:intersection_norms} and Lemma~\ref{lem:factorization-norms}. When $x \geq 0$, the corresponding map $u\in\lin(\matr_N,\matr_Q)$ is completely positive. This implies that the completely bounded norm of $u$, as defined in~\eqref{eq:choi-matrix} is given by $\| u \|_{\mathds{L}^\infty_{N}\to\mathds{L}^\infty_{N}} = \| u(\1) \|_{\infty}$. We also know that for completely positive map, we can find a representation $u(\beta) = \hat{a} \beta \otimes \1_{NQ} \hat{a}^{*}$. This implies that $\| x \|_{(\infty;1)} = \| u \|_{_{\mathds{L}^\infty_{N}\to\mathds{L}^\infty_{N}}} = \| \hat{a} \hat{a}^* \|_{\infty}$, and then we get
\begin{align}
\| x \|_{\capcap}\leq\left(\max\left\{ \sqrt{2^k} \| a \|_{\infty}, \| \hat{a} \|_{\infty} \right\}\right)^2=\max \left\{2^k\| x \|_{\infty}, \| x \|_{(1;\infty)} \right\}\ .
\end{align}
For $x \in \matr_N$, we perform a polar decomposition of $x$, $x = UP$ where $U$ is a unitary matrix and $P$ is positive semidefinite. Then let $a = U \sqrt{P}$, $b = \sqrt{P}$, and hence $\| a \|_{(\infty;\infty)} = \| b \|_{(\infty;\infty)} =\| P \|_{(\infty;\infty)}^{1/2}$. Moreover, we have $\| \Gamma(a) \|_{\infty} = \| a \|_2 = \| b \|_2 = \| \sqrt{P} \|_2 = \| P \|_1$, and we finally get $\max\left\{ 2^k \| x \|_{\infty}, \| x \|_{1} \right\} \leq \| x \|_{\capcap}$. 
\end{proof}

\begin{proposition}\label{prop:decomposition-positive}
For $x \in \matr_Q(\matr_N)$ we have
\begin{align}
x = x_1 - x_2 + i x_3 - i x_4\quad\mathrm{with}\quad x_i \geq 0\quad\mathrm{and}\quad\| x_i \|_{(\capcap)_Q} \leq \| x \|_{(\capcap)_Q}\ .
\end{align}
\end{proposition}

\begin{proof}
Let $x = ab^*$ be a factorization achieving the minimum of~\eqref{eq:def-cap-cap} with
\begin{align}
\max\left\{ \sqrt{2^k} \| a \|_{(\infty;\infty)}, \| \Gamma(a) \|_{\infty} \right\} = \max\left\{ \sqrt{2^k} \| b \|_{(\infty;\infty)}, \| \Gamma(b) \|_{\infty} \right\}\ ,
\end{align}
which can be achieved by scaling $a$ and $b$. We define $x_1 = \frac{1}{4} (a+b) (a+b)^*, x_2 = \frac{1}{4} (a - b) (a-b)^*, x_3 = \frac{1}{4} (a+ib) (a+ib)^*$ and $x_4 = \frac{1}{4}(a-ib)(a-ib)^*$. It is simple to verify that $x = x_1 - x_2 + i x_3 - i x_4$. For the norms, we have 
\begin{align}
\| x_1 \|_{(\capcap)_Q} 
&\leq \frac{1}{4} \max\left\{ \sqrt{2^k} \| a + b\|_{(\infty;\infty)}, \| \Gamma(a+b)\|_{\infty} \right\}^2 \\
&\leq \frac{1}{4} \max\left\{ \sqrt{2^k} (\| a \|_{(\infty;\infty)} + \| b \|_{(\infty;\infty)}), \| \Gamma(a) \|_{\infty} + \| \Gamma(b) \|_{\infty} \right\}^2 \\
&\leq \frac{1}{4} \left(\max\left\{ \sqrt{2^k} \| a \|_{(\infty;\infty)}, \| \Gamma(a) \|_{\infty} \right\} + \max\left\{ \sqrt{2^k} \| b \|_{(\infty;\infty)}, \| \Gamma(b) \|_{\infty} \right\}  \right)^2 \\
&= \| x \|_{(\capcap)_Q}\ .
\end{align}
The same argument also holds for $x_i$ with $i \in \{2,3,4\}$.
\end{proof}

We now have everything at hand that is relevant for extractors: the operator space $\capcap\{K\mathds{L}_N^\infty,\mathds{L}_N^1\}$ defined in~\eqref{eq:def-cap-cap} as the extension of the $\cap\{2^k \ell^\infty_N, \ell^1_N\}$-norm (for the input condition), and the trace class operator space $\mathds{L}_M^1$ with the defining norms $(\infty;1)$ as in~\eqref{eq:l1-os} as the extension of the $\ell^1_M$-norm (for the output condition).

\extToNormQuantum*

\begin{proof}
We first prove~\eqref{eq:quantum_ext1}. For $\rho_{QN}\in\cS(QN)$ with $H_{\min}(N|Q)_{\rho}\geq k$ we have $\|\rho_{QN}\|_{(1;\infty)} \leq 2^{-k}$ as well as $\|\rho_{QN}\|_{(1;1)} \leq 1$. Hence, we get for $\sigma_{Q}\in\cS(Q)$ that
\begin{align}
\| \Delta[\ext](\rho_{QN})\|_{(1;1)} &= \left\| \sigma^{1/2}_Q\Delta[\ext]\left(\sigma^{-1/2}_Q \rho_{QN}\sigma^{-1/2}_Q\right)\sigma^{1/2}_Q\right\|_{(1;1)}\label{eq:cb_extractor1}\\
&\leq \left\| \Delta[\ext]\left(\sigma^{-1/2}_Q \rho_{QN} \sigma^{-1/2}_Q\right) \right\|_{(\infty;1)} \\
&\leq \| \Delta[\ext] \|_{\capcap\left(2^k\mathds{L}_N^\infty,\mathds{N}_N^1\right)\to\mathds{L}_M^1} \left\| \sigma^{-1/2}_Q \rho_{QN} \sigma_Q^{-1/2} \right\|_{\capcap} \\
&= \| \Delta[\ext] \|_{\capcap\left(2^k\mathds{L}_N^\infty,\mathds{N}_N^1\right)\to\mathds{L}_M^1} \cdot \max\left\{ 2^k \left\| \sigma^{-1/2}_Q \rho_{QN} \sigma^{-1/2}_Q \right\|_{(\infty;\infty)},\left\|\sigma^{-1/2}_Q \rho_{QN} \sigma^{-1/2}_Q \right\|_{(\infty;1)}\right\}\ ,
\end{align}
where we used the simple expression for the $\capcap$-norm of positive operators (Proposition~\ref{prop:new-old-norm}). We now apply the previous inequality to
\begin{align}
\text{$\sigma_{Q}:=\frac{\omega_{Q}+\rho_Q}{2}$ with $\rho_{QN} \leq 2^{-k} \omega_{Q} \otimes \1_N$ and $\omega_{Q}\in\cS(Q)$.}
\end{align}
Then, we have $2^k \rho_{QN} \leq \omega_Q \otimes \1_N \leq 2 \sigma_Q \otimes \1_N$, which means
\begin{align}
2^{k}\left\|\sigma^{-1/2}_Q \rho_{QN} \sigma^{-1/2}_Q\right\|_{(\infty;\infty)} \leq 2\ .
\end{align}
In addition, as $\rho_Q \leq 2 \sigma$ we get 
\begin{align}
\left\| \sigma^{-1/2}_Q \rho_{QN} \sigma^{-1/2}_Q \right\|_{(\infty;1)}= \left\| \sigma^{-1/2}_Q \rho_Q \sigma^{-1/2}_Q \right\|_{\infty}\leq 2\ .\label{eq:cb_extractor2}
\end{align}
Taken~\eqref{eq:cb_extractor1}--\eqref{eq:cb_extractor2} together proves~\eqref{eq:quantum_ext1}.

We now prove~\eqref{eq:quantum_ext2}. Let $z\in\matr_Q(\mathbb{C}^N)$ with $\| z \|_{\capcap} \leq 1$. By definition of the $\|\cdot\|_{\capcap}$ norm (Definition~\ref{def:intersection_norms}), there exists a factorization $z = a b^*$ such that
\begin{align}
\max\left\{\sqrt{2^{k}}\|a\|_{(\infty; \infty)}, \| \Gamma(a) \|_{\infty} \right\}=\max \left\{ \sqrt{2^{k}} \| b \|_{(\infty; \infty)},\|\Gamma(b) \|_{\infty} \right\} \leq 1\ .
\end{align}
We then define $\hat{z}\in\matr_{2QN}$ as
\begin{align}\label{eq:zhat}
\hat{z}\eqdef\begin{pmatrix}
aa^* & b a^* \\
a b^* & b b^*
\end{pmatrix}
=\left( a \otimes \ket{0} \bra{0} + b \otimes \ket{1} \bra{0}  \right) \cdot \left( a \otimes \ket{0} \bra{0} + b \otimes \ket{1} \bra{0}  \right)^*\ ,
\end{align}
and estimate
\begin{align}
\| \hat{z} \|_{\capcap} &\leq\left(\max\left\{ \sqrt{2^{k}} \| a \otimes \proj{0} + b \otimes \ket{0} \bra{1}  \|_{(\infty;\infty)}, \| \Gamma(a) \otimes \proj{0} + \Gamma(b) \otimes \ket{0} \bra{1} \|_{\infty}\right\}\right)^2 \\
&\leq\left(2\max\left\{ \sqrt{2^{k}} \| a \|_{(\infty; \infty)}, \| \Gamma(a) \|_{\infty}\right\}\right)^2 \\
&= 4 \| z \|_{\capcap}\\
&\leq4\ .
\end{align}
By~\eqref{eq:zhat} we have $\hat{z}\geq 0$ and Proposition~\ref{prop:new-old-norm} then implies
\begin{align}\label{eq:zhat_action}
\max \left\{ 2^{k} \| \hat{z} \|_{(\infty;\infty)}, \| \hat{z} \|_{(\infty;1)} \right\}=\| \hat{z} \|_{\cap \cdot \cap}\leq4\ .
\end{align}
Now, we evaluate
\begin{align}\label{eq:Uopt1}
\| \Delta[\ext](z) \|_{(\infty;1)}= \sup_{\stackrel{c,d \in \matr_Q}{\| c \|_2 \leq 1, \| d \|_2 \leq 1}}\big\| c \Delta[\ext](z) d \big\|_{(1;1)}&= \sup_{U, c,d} \Big| \tr\big[c \Delta[\ext](z) d \cdot U\big] \Big| \\
&= \sup_{U,c,d} \left|  \tr\left[ \Delta[\ext]\left(
\begin{pmatrix} 
c & 0 \\
0 & c 
\end{pmatrix}
\hat{z} 
\begin{pmatrix}
d & 0 \\
0 & d 
\end{pmatrix}\right)  \cdot U \otimes \ket{1} \bra{0} \right] \right| \ .\label{eq:Uopt2}
\end{align}
For the positive semidefinite operator 
\begin{align}
\rho:=\frac{1}{8}\begin{pmatrix}
c & 0 \\
0 & c 
\end{pmatrix}
\hat{z} 
\begin{pmatrix}
d & 0 \\
0 & d 
\end{pmatrix}
\end{align}
we get by the definition of the $(1;\infty)$-norm in~\eqref{eq:def-one-infty} as well as~\eqref{eq:zhat_action} that
\begin{align}
\| \rho \|_{(1;\infty)} \leq \frac{1}{8} \cdot 2 \| \hat{z} \|_{(\infty;\infty)} \leq2^{-k}\quad\mathrm{and}\quad\| \rho \|_{(1;1)} \leq \frac{1}{8} \cdot 2 \| \hat{z} \|_{(\infty;1)} \leq 1\ .
\end{align}
In order to have a valid state, we define
\begin{align}
\bar{\rho}:=\rho+\big(1-\tr[\rho]\big)\frac{\1}{QN}\quad\mathrm{with}\quad\|\bar{\rho}\|_{(1;\infty)} \leq2^{-k} + \frac{1}{N} \leq2^{-(k-1)}\ .
\end{align}
Now, by assumption $\ext$ is a quantum-proof $(k-1,\eps)$-extractor and with~\eqref{eq:Uopt1}--\eqref{eq:Uopt2} we conclude that
\begin{align}
\| \Delta[\ext](z) \|_{(\infty;1)}\leq8\cdot\| \Delta[\ext](\rho) \|_{(1;1)}\leq8\cdot\| \Delta[\ext](\bar{\rho}) \|_{(1;1)}\leq 8 \eps\ .
\end{align}
\end{proof}


\subsection{Stability bounds}

This way of writing the extractor and the quantum-proof extractor conditions allows us to use tools from operator space theory to compare the two concepts. As a first straightforward application, we can use dimension dependent bounds that are known for the maximal possible ratios between the completely bounded norm and the bounded norm.

\begin{corollary}\label{cor:out-size}
Every $(k,\eps)$-extractor is a quantum-proof $(k+1,8\sqrt{2^{m}}\eps)$-extractor.
\end{corollary}

Observe that this result is only interesting when $m$ is small. In particular, it is only useful for weak extractors, as strong extractors have $2^m \geq 2^d = \Omega(\eps^{-2})$ (by the converse bound~\eqref{eq:extractor_converse}).

\begin{proof}
By Proposition~\ref{prop:ext-to-norm} we have that $\| \Delta[\ext] : \capcap\{2^{k+1} \ell^\infty_N, \ell^1_N\} \to \ell^1_M \| \leq 4 \eps$. Now, we estimate using~\cite[Theorem 3.8]{Pis03},
\begin{align}
\left\|\Delta[\ext]:\capcap\{2^{k+1}\mathds{L}^\infty_N,\mathds{L}^1_N\}\to\mathds{L}^1_M\right\|_{\cb}&\leq\left\| \Delta[\ext]\right\|_{\capcap\{2^{k+1} \mathds{L}^\infty_N,\mathds{L}^1_N\} \to(\mathds{L}^1_M)^{\min}}\cdot\left\|\1\right\|_{(\mathds{L}^1_M)^{\min}\to\mathds{L}^1_M}\\
&=\left\| \Delta[\ext]\right\|_{\capcap\{2^{k+1} \ell^\infty_N,\ell^1_N\} \to\ell^1_M}\cdot\left\|\1\right\|_{\mathds{L}^{\infty}_M\to(\mathds{L}^\infty_M)^{\max}}\\
&\leq 4 \eps \sqrt{2^{m}}\ .
\end{align}
The claim then follows by Theorem~\ref{thm:ext-to-norm-quantum}.
\end{proof}

\smallerrorent*

\begin{proof}
By operator space duality for the trace class operator space (Proposition~\ref{prop:s1-os}) and Proposition~\ref{prop:decomposition-positive} we get
\begin{align}
\left\|\Delta[\ext]_{S} : \capcap\{2^{k+\log(1/\eps)} \mathds{L}^\infty_N,\mathds{L}^1_N\}\to\mathds{L}^1_M\right\|_{\cb}\leq4\sup\left\{\left\|\E_s \sum_y \sum_x \left(\delta_{f_s(x) = y} - 2^{-m}\right)p(x) \otimes q(s,y)\right\|_{(\infty;\infty)}\right\}\ ,
\end{align}
where the supremum is over all
\begin{align}
p(x)\in\matr_Q\quad&\mathrm{with}\quad0\leq p(x)\leq\eps2^{-k}\1\,,\,\sum_x p(x) \leq \1\\
q(s,y)\in\matr_Q\quad&\mathrm{with}\quad\norm{q(s,y)}_{\infty} \leq1
\end{align}
and $Q\in\mathbb{N}$. We now apply the operator version of the Cauchy-Schwarz inequality due to Haagerup~\cite[Chapter 7]{Pis03},
\begin{align}
&\left\|\E_s \sum_y \left(\sum_x \left(\delta_{f_s(x) = y} - 2^{-m'}\right)p(x)\right) \otimes q(s,y)\right\|_{(\infty;\infty)} \notag\\
&\leq\left\|\E_s \sum_y \left(\sum_x \left(\delta_{f_s(x) = y} - 2^{-m'}\right)p(x)\right) \otimes \overline{\left(\sum_{x'} \left(\delta_{f_s(x') = y} - 2^{-m'}\right)p(x')\right)}\right\|^{1/2}_{(\infty;\infty)}\left\|\E_s \sum_y q(s,y) \otimes \overline{q(s,y)}\right\|^{1/2}_{(\infty;\infty)}\ .
\end{align}
The second term is upper bounded by
\begin{align}
\left\|\E_s \sum_y q(s,y) \otimes \overline{q(s,y)}\right\|_{(\infty;\infty)} \leq \E_s \sum_y \norm{q(s,y)}_{\infty}^2 \leq\sqrt{2^{m'}}\ ,
\end{align}
and we are left with the first term. The operator whose $(\infty;\infty)$-norm has to be estimated is hermitian, and hence we arrive via norm duality at
\begin{align}
&\sup_{\stackrel{C\in\matr_Q}{\tr[C^*C]\leq1}}\left|\,\E_s \sum_y \sum_{x,x'} \left(\delta_{f_s(x) = y} - 2^{-m'}\right)\tr\left[C p(x) C^* p(x')\right]\left(\delta_{f_s(x') = y} - 2^{-m'}\right)\right| \notag\\
&\leq\sup_{\stackrel{C\in\matr_Q}{\tr[C^*C]\leq1}}\E_s \sum_y \sum_{x'}\left| \sum_{x}\left(\delta_{f_s(x) = y} - 2^{-m'}\right) l_{xx'}(C)  \right|\ ,
\end{align}
where \(l_{xx'}(C):=\tr[C p(x) C^* p(x')] \) is a positive function on \(N \times N\) with
\begin{align}
\sum_{x,x'} l_{xx'}(C) = \tr\left[C \sum_x p(x) C^* \sum_{x'} p(x')\right] \leq 1\quad\mathrm{and}\quad l_{xx'}(C)\leq\eps2^{-k}\,\tr[CC^* p(x')]\ .
\end{align}
Hence, the distribution $l_{xx'}(C)$ has conditional min-entropy $H_{\min}(X|X')_{l(C)}\geq k+\log(1/\eps)$, and by Markov's inequality
\begin{align}
\pr{H_{\min}(X|X'=x')_{l(C)}\leq k}\leq\eps\ .
\end{align}
Finally, we get by the assumption that the bounded norm of the extractor is upper bounded by $\eps$ that
\begin{align}
\E_s \sum_y \sum_{x'}\left| \sum_{x}\left(\delta_{f_s(x) = y} - 2^{-m'}\right)l_{xx'}(C)\right|&=\E_{x'}\E_s\sum_{y}\Big|\sum_{x}\left(\delta_{f_s(x) = y} - 2^{-m'}\right)l_{x|x'}(C)\Big|\\
&\leq2\eps\ .
\end{align}
By putting everything together we get the upper bound $4\sqrt{2^{m'}}\sqrt{2\eps}$ on the completely bounded norm of the extractor as claimed.
\end{proof}

An interesting application is for very high min-entropy extractors.

\highMinEntExt*

\begin{proof}
We have for any distribution $P$,
\begin{align}
\text{$\|P\|_{2^{k}\ell^\infty_N}\leq\|P\|_{\cap(2^{k}\ell^\infty_N,\ell^1_N)}\leq2^{n-k}\|P\|_{2^{k}\ell^\infty_N}$, and this implies that $\left\|\Delta[\ext]:2^{k+1}\ell^\infty_N\to\ell^1_N\right\|\leq2^{n-k}\eps$.}
\end{align}
By Grothendieck's inequality~\cite[Corollary 14.2]{Pis12}, we conclude
\begin{align}
\left\|\Delta[\ext]:\capcap\{2^{k}\mathds{L}^\infty_N,\mathds{L}^1_N\}\to\mathds{L}^1_M\right\|_{\cb}\leq\left\|\Delta[\ext]:2^{k}\mathds{L}^\infty_N\to\mathds{L}^1_M\right\|_{\cb}\leq K_G2^{n-k}\eps\ .
\end{align}
\end{proof}


\section{Condensers}\label{sec:condenser}

\subsection{Condenser property as a bounded norm}

\condToNorm*

\begin{proof}
We first prove~\eqref{eq:condenser_1}. Let $P$ be a distribution with min-entropy at least $k$. Then, we have $\| P \|_{\cap\{2^k\ell^\infty_N,\ell^1_N\}}\leq1$ and this implies $\|[\cond](P)\|_{\Sigma\{2^{k'}\ell^\infty_M,\ell^1_M\}}\leq\eps$. Hence, there is a decomposition $[\cond](P)=Q_1+Q_2$ such that $\|Q_1\|_{2^{k'}\ell^\infty_M}\leq\eps$ and $\|Q_2\|_{\ell^1_M}=\|[\cond](P)-Q_1\|_{\ell^1_M}\leq\eps$. This is almost another way of saying $H_{\min}^{\eps}([\cond](P))\geq k'+\log(1/\eps)$ except for the fact that $Q_1$ and $Q_2$ might have negative entries. We now show that we may construct nonnegative $Q_1''$ and $Q_2''$ with the same properties. In fact, consider $Q'_1 = \max(Q_1, 0)$ and $Q'_2 = [\cond](P) - Q'_1 = \min([\cond](P),Q_2)$. Then, we still have $Q'_1 + Q'_2 = [\cond](P)$, and in addition $\| Q'_1 \|_{2^{k'}\ell^\infty_M} \leq \|Q_1\|_{2^{k'}\ell^\infty_M}$ as well as $\|Q'_2\|_{\ell^1_M} \leq \|Q_2\|_{\ell^1_M}$. But we do not necessarily have $Q'_2$ nonnegative. For this, we can define $Q_1'' = \min(Q_1',P)$ and $Q_2'' = P - Q_1'' = \max(0, Q_2')$. Then we get $Q_1''$ and $Q_2''$ are nonnegative and $\|Q_1''\|_{\ell^\infty_M} \leq \eps 2^{-k'}$ and $\| P - Q_1''\|_{\ell^1_M} \leq \eps$.

To prove~\eqref{eq:condenser_2}, assume that the set of functions $\cond$ defines a $(k-1)\to_{\eps}k'+\log(1/\eps)$ condenser and consider $P$ such that $\|P\|_{\ell^1_N}\leq1$ and $\|P\|_{\ell^\infty_N} \leq 2^{-k}$. Let $P^+(x) = \max\{P(x), 0\}$ and $P^-(x) = \max\{-P(x), 0\}$, and note that $\| P^+ \|_{\ell^\infty_N}, \| P^{-} \|_{\ell^\infty_N} \leq 2^{-k}$. We extend $P^+$ into a probability distribution $\bar{P}^+ = P^+ + (1-\|P^+\|_{\ell^1_N})\upsilon_N$ and similarly $\bar{P}^-(x) = P^-(x) +(1-\|P^-\|_{\ell^1_N}) \upsilon_N$. Now observe that $\| \bar{P}^+ \|_{\ell^\infty_N} \leq \| P \|_{\ell^\infty_N} + \frac{1}{N} \leq 2^{-(k-1)}$ and similarly for $\| \bar{P}^+ \|_{\ell^\infty_N}$. Thus, we have that
\begin{align}
\|[\cond](P) \|_{ \Sigma\{2^{k'}\ell^\infty_M, \ell^1_M\}} 
&= \| [\cond](P^+) - [\cond](P^-) \|_{ \Sigma\{2^{k'}\ell^\infty_M, \ell^1_M\}} \\
&\leq \| [\cond](P^+) \|_{ \Sigma\{2^{k'}\ell^\infty_M, \ell^1_M\}}+\| [\cond](P^-) \|_{ \Sigma\{2^{k'}\ell^\infty_M, \ell^1_M\}} \\
&\leq \| [\cond](\bar{P}^+) \|_{ \Sigma\{2^{k'}\ell^\infty_M, \ell^1_M\}} + (1-\|P^+\|_{\ell_1}) \| [\cond](\upsilon_N) \|_{ \Sigma\{2^{k'}\ell^\infty_M, \ell^1_M\}}\notag\\
&\quad+\| [\cond](\bar{P}^-) \|_{ \Sigma\{2^{k'}\ell^\infty_M, \ell^1_M\}} + (1-\|P^-\|_{\ell_1}) \| [\cond](\upsilon_N) \|_{ \Sigma\{2^{k'}\ell^\infty_M, \ell^1_M\}} \\
&\leq \| [\cond](\bar{P}^+) \|_{ \Sigma\{2^{k'}\ell^\infty_M, \ell^1_M\}} + \| [\cond](\bar{P}^-) \|_{ \Sigma\{2^{k'}\ell^\infty_M, \ell^1_M\}} + 2\| [\cond](\upsilon_N) \|_{ \Sigma\{2^{k'}\ell^\infty, \ell^1_M\}}\ .
\end{align}
Now the distributions $\bar{P}^+, \bar{P}^-$ and $\upsilon_N$ have min-entropy at least $k-1$. Hence there exists $\bar{Q}^+$ with $\|[\cond](\bar{P}^+)-\bar{Q}^+\|_{\ell^1_M}\leq\eps$ and $\| \bar{Q}^+ \|_{\ell^\infty_M}\leq\eps 2^{-k'}$. As a result, $\| [\cond](\bar{P}^+) \|_{ \Sigma\{2^{k'}\ell^\infty_M, \ell^1_M\}} \leq 2 \eps$. This implies $\|[\cond](P)\|_{\Sigma\{2^{k'}\ell^\infty_M, \ell^1_M\}}\leq 8 \eps$, which proves the desired result.
\end{proof}


\subsection{Quantum-proof as a completely bounded norm}

As we saw in Proposition~\ref{prop:ext-to-norm}, a condenser maps the \(\cap\)-normed space to its dual space. Since we expect the same to happen in the quantum-proof case, it is useful to have an understanding about the operator space dual of \(\capcap\). By expressing the \(\capcap\)-operator space using the Haagerup tensor product, the operator dual is easily identified. However, we do not want to elaborate further on this, since we will just use a simple estimate (see Lemma~\ref{lem:sigmanormforpos} below), and refer to Appendix~\ref{app:intersection} for the exact characterization. We again use a shorthand notation and denote the operator space dual of \(\capcap\) by $\sigsig$.

Like for the case of extractors, we are naturally interested in the case of matrices which are diagonal with respect to the first system, that is, elements of $\matr_Q(\mathbb{C}^N)\subset\matr_Q(\matr_N)$. The norms $\capcap$ and the $\sigsig$ are then defined on $\matr_Q(\mathbb{C}^N)$ via the embedding of $\matr_Q(\mathbb{C}^N)$ into $\matr_{QN}$ as block-diagonal matrices.

For positive $x\in\matr_Q\left(\mathbb{C}^N\right)$, the operator space dual norm $\sigsig$ has the following simple estimate.

\begin{lemma}\label{lem:sigmanormforpos}
Let $x = \sum_{j} x(j) \otimes \proj{i} \in \matr_Q\left(\mathbb{C}^N\right)$ be positive. Then, we have
\begin{align}
\half \, \norm{x}_{(\sigsig)_Q}\leq\inf\left\{\left\|2^{-k}\sum_j A(j) + B\right\|_\infty \;:\; x(j) \leq A(j) + B, \, A(j) \geq 0, \, B \geq 0 \right\} \leq \norm{x}_{(\sigsig)_Q}\ .
\end{align}
\end{lemma}

\begin{proof}
By the definition of the operator space dual, we have
\begin{align}\label{eq:mq-sigsig}
\| x \|_{(\sigsig)_Q} = \sup\left\{\left\|\sum_j x(j) \otimes y(j)\right\|_{(\infty;\infty)}\;:\; \left\|\sum_{j} y(j) \otimes \proj{j}\right\|_{(\capcap)_Q} \leq 1 \right\}\ .
\end{align}
Let us now decompose $y$ according to Proposition~\ref{prop:decomposition-positive}, $y = y_1 - y_2 + iy_3 - i y_4$. Then, we observe that in the maximization in \eqref{eq:mq-sigsig}, up to an additional factor of two we can assume that $y$ is also Hermitian. This implies that we can take $y_3 = y_4 = 0$. But then using the fact that $x \geq 0$, we have
\begin{align}
\left\| \sum_{j} x(j) \otimes y(j) \right\|_{(\infty;\infty)} \leq \left\| \sum_j x(j) \otimes y_1(j)\right\|_{(\infty;\infty)}\ ,
\end{align}
and thus we can assume that $y \geq 0$ in the further study of \eqref{eq:mq-sigsig}. Recalling that when $y \geq 0$,
\begin{align}
\| y \|_{(\capcap)_Q} = \| y \|_{\cap_Q} = \max\left\{2^k\| y \|_{\infty} , \left\| \sum_{j} y(j) \right\|_{\infty}\right\}\ ,
\end{align}
we have
\begin{align}
&\sup\left\{ \bra{\psi} \sum_j x(j) \otimes y(j)\ket{\psi} \;:\; 0 \leq y(j) \leq 2^{-k}\1,\, \sum_j y(j) \leq \1, \, \| \ket{\psi} \|_2 \leq 1 \right\}\notag\\
&= \sup\left\{  \sum_{j,l,l'} \lambda^{*}_{l} \lambda_{l} \bra{u_l} x(j) \ket{u_{l'}} \bra{v_l} y(j) \ket{v_{l'}} \;:\; 0 \leq y(j) \leq2^{-k}\1,\, \sum_j y(j) \leq \1 \, \| \ket{\psi} \|_2 \leq 1 \right\}\ ,
\end{align}
where $\ket{\psi} = \sum_{l} \lambda_l \ket{u_l} \ket{v_l}$ is a Schmidt decomposition. Now using the fact that $\bra{v_l} y(j) \ket{v_{l'}} = \bra{v_{l'}} y(j)^T \ket{v_{l}}$ and writing $C = \sum_{l} \lambda_l \ket{v_l} \bra{u_l}$, we can write
\begin{align}
\| x \|_{(\sigsig)_Q}=\sup\left\{\tr\left[\sum_j x(j) C^* y(j)^T C \right]\;:\; 0 \leq y(j) \leq2^{-k}\1,\,\sum_j y(j) \leq \1,\, \tr[C^* C] \leq 1 \right\}\ .
\end{align}
Moreover, the transposition leaves the operator inequalities involving the \(y(j)\)'s invariant, we can further simplify to (set $\sigma = C^* C$)
\begin{align}\label{eq:sigsig-sdp}
\sup\left\{\tr\left[\sum_j x(j) \hat{y}(j) \right]\;:\; 0 \leq \hat{y}(j) \leq2^{-k}\sigma,\, \sum_j \hat{y}(j) \leq \sigma,\, \tr[\sigma] \leq 1,\,\sigma \geq 0 \right\}\ .
\end{align}
This is an SDP and its dual is given by
\begin{align}\label{eq:sigsig-sdp-dual}
\inf\left\{\left\|2^{-k}\sum_j A(j) + B\right\|_\infty \;:\; x(j) \leq A(j) + B, \, A(j) \geq 0, \, B \geq 0 \right\}\ .
\end{align}
The program~\eqref{eq:sigsig-sdp} is clearly feasible and~\eqref{eq:sigsig-sdp-dual} is strictly feasible and thus strong duality is satisfied, i.e., the programs~\eqref{eq:sigsig-sdp} and~\eqref{eq:sigsig-sdp-dual} have the same value.
\end{proof}

Such kind of intersection norms for operator spaces have been extensively employed by Junge {\it et al.} in their study of non-commutative $L_p$-spaces and their relation to free probability, see for instance the monograph~\cite{Junge:2010cg}. We expect that many of Junge and his co-workers' techniques are applicable to questions regarding the stability of pseudorandom objects and hope that our work serves as a starting point for such kind of investigations.

We first need a lemma relating the dual norm of $\capcap$ for positive elements to the smooth conditional min-entropy.

\begin{lemma}\label{lem:smoothmin-sigmanorm}
For $\rho \in \matr_{QN}$ positive we have
\begin{align}
2^{k}\|\rho\|_{(\sigsig)_Q}\leq\eps\quad&\Rightarrow\quad H_{\min}^{2\sqrt{\eps}}(N|Q)_\rho \geq k + \log(1/\eps)\label{eq:smooth_min1}\\
H_{\min}^{\eps}(N|Q) \geq k + \log(1/\eps)\quad&\Rightarrow\quad2^{k}\|\rho\|_{(\sigsig)_Q}\leq2\eps\ .\label{eq:smooth_min2}
\end{align}
\end{lemma}

\begin{proof}
We first prove~\eqref{eq:smooth_min1}. Due to the fact that $\rho$ is positive, the optimal decomposition of $x$ with respect to Proposition \ref{prop:decomposition-positive} is one where only the positive term is non-zero. Hence, we have due to Proposition~\ref{prop:new-old-norm},
\begin{align}
\|\rho\|_{(\sigsig)_Q}=\sup\big\{ |\tr[\rho x]| \,:\, \norm{x}_{(\capcap)_Q} \leq 1 \big\} &= \sup\big\{ |\tr[\rho x]| \,:\, x \geq 0,\,\norm{x}_{\cap_Q} \leq 1\big\}\\
&= 2^{-k}\,\sup\left\{ \tr[\rho x] \,:\,x \geq 0,\, \tr_N(x) \leq 2^k\1, x \leq \1\right\}\ .
\end{align}
This is a semidefinite program, and by strong duality we get 
\begin{align}
\sup\left\{ \tr[\rho x] \,:\,x \geq 0,\,\tr_N(x) \leq 2^k\1, x \leq \1\right\} =\inf\left\{\tr[A_1] + 2^{k}\tr[A_2] \,:\, \rho \leq A_1 + \1_N \otimes A_2, \; A_{1,2} \geq 0\right\}\ .
\end{align}
Hence, we find two positive matrices $A_1$ and $A_2$ such that $\tr[A_1] \leq \eps$, and $\tr[A_2] \leq \eps \, 2^{-k}$. We now consider
\begin{align}
\hat{\rho}:=B^* B\quad\mathrm{with}\quad B:=\rho^{1/2} (A_1 + \1 \otimes A_2)^{-1/2} \1 \otimes A_2^{1/2}\ .
\end{align}
We have that its trace is smaller than one,
\begin{align}
\tr[\hat{\rho}] &= \tr\left[\1 \otimes A_2^{1/2} (A_1 + \1 \otimes A_2)^{-1/2} \rho (A_1 + \1 \otimes A_2)^{-1/2} \1 \otimes A_2^{1/2}\right]\\
&= \tr\left[\rho (A_1 + \1 \otimes A_2)^{-1/2} \1 \otimes A_2 (A_1 + \1 \otimes A_2)^{-1/2}\right]\\
&\leq \tr[\rho]\\
&= 1\ ,
\end{align} 
since $(A_1 + \1 \otimes A_2)^{-1/2} \1 \otimes A_2 (A_1 + \1 \otimes A_2)^{-1/2} \leq (A_1 + \1 \otimes A_2)^{-1/2} A_1 + \1 \otimes A_2 (A_1 + \1 \otimes A_2)^{-1/2} = \1$. This implies that $\hat{\rho}$ is a sub-normalized state, and its min-entropy is as least $k + \log(1/\eps)$ since we have $(A_1 + \1 \otimes A_2)^{-1/2} \rho (A_1 + \1 \otimes A_2)^{-1/2} \leq \1$ by construction and $\tr[A_2] \leq \eps \, 2^{-k}$. Simple rescaling of $A_2$ makes it into a density matrix, and the corresponding factor is picked up by the inner term. Moreover, let us consider its trace norm distance to $\rho$. First, we have using the bounded trace of $\rho$ and $B^* B$,
\begin{align}
\left\|\rho - B^*B\right\|_1 =\left\|(\rho^{1/2} - B^*)\rho^{1/2} + B^*(\rho^{1/2} - B)\right\|_1 \leq2\left\|\rho^{1/2} - B\right\|_2\ .
\end{align}
We then estimate further, using the Powers-Stoermer inequality in the last step,
\begin{align}
\norm{\rho^{1/2} - B}_2 &=\left\|\rho^{1/2} (A_1 + \1 \otimes A_2)^{-1/2} (A_1 + \1 \otimes A_2)^{1/2} - \rho^{1/2} (A_1 + \1 \otimes A_2)^{-1/2} \1 \otimes A_2^{1/2}\right\|_2\\
&\leq\left\|\rho^{1/2} (A_1 + \1 \otimes A_2)^{-1/2}\right\|_\infty\cdot\left\|(A_1 + \1 \otimes A_2)^{1/2} - \1 \otimes A_2^{1/2}\right\|_2\\
&\leq\sqrt{\norm{A_1}_1}\\
&= \sqrt{\eps}\ ,
\end{align}
and the first statement follows.
	
To prove~\eqref{eq:smooth_min2}, let $\hat{\rho}$ be the state which achieves the smooth min-entropy, i.e., we have $\hat{\rho} \leq \eps\,2^{-k}\1 \otimes \sigma$ and $\norm{\hat{\rho} - \rho}_1 \leq \eps$. Let $x$ be now chosen such that $\tr_N(x) \leq 2^k\1$, $x \leq \1$. We then have
\begin{align}
2^{-k}\tr[x \rho] = \tr[x \hat{\rho}] + \tr[x (\hat{\rho} - \rho)] \leq \eps \, 2^{-k} \tr\big[\sigma \tr_N[x]\big] + \eps\,\norm{x}_\infty\leq2\cdot2^{-k}\eps\ ,
\end{align} 
and the assertion is proven.
\end{proof}

Expressing the smooth conditional min-entropy in terms of the $\sigsig$-norm (Lemma~\ref{lem:smoothmin-sigmanorm}) allows us to characterize quantum-proof condensers using operator space norms.

\condToNormQuantum*

\begin{proof}
We first prove~\eqref{eq:quantumcondenser1}. Let $\rho_{NQ}\in\matr_Q\left(\mathbb{C}^N\right)$ be a state with conditional min-entropy at least $k$, and let $\omega \in \matr_Q$ such that $\rho_{NQ} \leq 2^{-k} \omega \otimes \1_N$. We proceed as in the proof of Theorem~\ref{thm:ext-to-norm-quantum} and define $\sigma = \frac{\rho_Q + \omega}{2}$. It follows that $\left\|\sigma^{-1/2} \rho \sigma^{-1/2}\right\|_{(\capcap)_Q}\leq 2$ and hence we have due to Lemma~\ref{lem:sigmanormforpos} that there are positive matrices $A \in \matr_Q\left(\mathbb{C}^N\right)$, $B\in\matr_Q$ such that
\begin{align}
\cond(\rho)=\sigma^{1/2} \cond\left(\sigma^{-1/2} \rho \sigma^{-1/2}\right) \sigma^{1/2}\leq \sigma^{1/2}\left(A + \1 \otimes B\right) \sigma^{1/2}\quad\mathrm{and}\quad\left(\sum_{y \in M} A(y) + 2^{k'} \,B\right) \leq \eps \, \1\ .
\end{align}
Now take $x\in\matr_Q\left(\mathbb{C}^M\right)$, fulfilling $x\geq0$, $x\leq\1$ and $\sum_{y\in M}x(y)\leq2^{k'}\1$. Then, we have
\begin{align}
\tr[\rho x]\leq \tr\left[\sigma^{1/2} (A + \1_M \otimes B) \sigma^{1/2} x\right]&\leq \tr_Q\left[\sigma^{1/2} \sum_{y \in M} A(y) \sigma^{1/2}\right] + \tr_Q\left[\sigma^{1/2} B \sigma^{1/2} \sum_{y \in M} x(y)\right]\\
&\leq\tr\left[\sigma\left(\sum_{y \in M} A(y) + 2^{k'} \,B\right)\right]\\
&\leq\eps\ .
\end{align}
and~\eqref{eq:quantumcondenser1} follows by Lemma~\ref{lem:smoothmin-sigmanorm}. 

For proving~\eqref{eq:quantumcondenser2}, take $x\in\matr_Q\left(\mathbb{C}^N\right)$ with $\norm{x}_{(\capcap)_Q} \leq 1$. Due to Proposition~\ref{prop:decomposition-positive} we can find a decomposition of $x=x_1-x_2+i(x_3-x_4)$ into positive terms fulfilling the same norm constraint. We will proceed to bound 
\begin{align}
\text{$\left\|\sum_{y\in M}z(y)\otimes\cond(x_i)(y)\right\|_{(\infty;\infty)}\leq\eps$, where $\norm{z}_{\capcap\{2^{k'}\ell^\infty_M,\ell^1_M\}} \leq 2^{k'}$.}
\end{align} 
This implies the assertion by operator space duality. Due to the fact that $\cond(x_i)(y)$ is positive, the optimal decomposition of $z$ with respect to Proposition~\ref{prop:decomposition-positive} is one where only the positive term is non-zero. Hence, it is enough to bound the expression above for positive $z$,
\begin{align}
\left\|\sum_{y \in M} z(y) \otimes \cond(x_i)(y)\right\|_\infty= \sup_{\ket{\psi} \in \C^{Q^2}} \bra{\psi} \sum_{y \in M} z(y) \otimes \cond(x_i)(y) \ket{\psi}= \sup_{C}\tr\left[\sum_{y \in M} C z(y)^T C^*\, \cond(x_i)(y)\right]\ ,
\end{align}
where the last supremum is over all operators $C$ on $Q$ with singular values equal to the Schmidt coefficients of $\psi$. Due to the properties of $x_i$,
\begin{align}
\sum_{k \in N} \tr_Q\left[C^* x_i(k) C\right] \leq 1\quad\mathrm{and}\quad C^* x_i(y) C \leq 2^{-k} C^* C\ ,
\end{align}
which implies that $\rho_i(k) = C^* x_i(k) C$ defines a sub-normalized classical-quantum state on $NQ$ having min-entropy at least $k$. In order to have a valid state, we define
\begin{align}
\bar{\rho_i} = \rho_i + (1-\tr[\rho_i]) \frac{\1}{Q} \otimes \frac{\1}{N}\ .
\end{align}
As before, we can say that $\| \bar{\rho} \|_{(1;\infty)}\leq2^{-k}+\frac{1}{N}\leq2^{-k+1}$ and hence the min-entropy is at least $k-1$. It follows that $\cond(C^* x_i C) \leq \cond{\bar{\rho_i}} = \omega_i$ where $\omega_i$ is a state with $\eps$-smooth min-entropy at least $k'$. Hence, we find by Lemma~\ref{lem:smoothmin-sigmanorm}
\begin{align}
\tr\left[\sum_{y \in M} C z(y)^T C^* \, \cond(x_i)(y)\right]= \tr\left[\sum_{y \in M} z(y^T) \, \cond(C^* x_i C)(y)\right] \,\leq\, \tr\left[\sum_{y \in M} z(y)^T \omega_i(y)\right] \,\leq \, 2 \,\eps\ .
\end{align}
\end{proof}


\subsection{Graph Theory}\label{sec:graph_theory}

\bipartitegraph*

\begin{proof}
Using the fact that $\Sigma\{2^{k'}\ell^\infty_M,\ell^1_M\}$ is the dual norm of the norm $\cap\{2^{-k'}\ell^1_M,\ell^\infty_M\}$, we write 
\begin{align}
&\left\|\cond:\cap\{2^{k}\ell^\infty_N, \ell^1_N\} \to \Sigma\{2^{k'}\ell^\infty_M, \ell^1_M\}\right\|\notag\\
&= \max\left\{\left|\sum_{y \in M} \cond(f)_y \cdot g_y\right| \;:\; \norm{f}_{\cap\{ 2^k \ell^\infty_N, \ell^1_N\}} \leq 1 ,\; \norm{g}_{\cap\{2^{-k'} \ell^1_M, \ell^\infty_M\}} \leq 1\,\right\} \\
&= \max\left\{  \left| \frac{1}{D} \sum_{x \in N, y \in M, s \in D} \delta_{\Gamma(x,s)=y} f_x \cdot g_y \right| \;:\; \norm{f}_{\cap\{ 2^k \ell^\infty_N, \ell^1_N\}} \leq 1 ,\; \norm{g}_{\cap\{2^{-k'} \ell^1_M, \ell^\infty_M\}} \leq 1\,  \right\} \\
&= \frac{1}{D \cdot 2^k} \max\left\{  \left| \sum_{x \in N, y \in M, s \in D} \delta_{\Gamma(x,s)=y} f_x \cdot g_y \right| \;:\; \norm{f}_{\cap\{ \ell^\infty_N, 2^{-k} \ell^1_N\}} \leq 1 ,\; \norm{g}_{\cap\{2^{-k'} \ell^1_M, \ell^\infty_M\}} \leq 1\,\right\}\ ,
\end{align}
where $f_x$, $g_y$ denote the components of the vectors $f \in \Rl^N$, $g \in \Rl^M$. Because the matrix elements of the tensor are all positive, we can restrict the maximization to vectors with positive entries. Then, the norm conditions on the vector $f$ translates into $0\leq f_x \leq 1$, $\sum_x f_x \leq2^{k}$. However, the extreme points of this convex sets are just the characteristic vectors of subsets of $N$ of size $2^{k}$. Due to the convex character of the objective function in both variables $f$ and $g$, we hence end up with the quadratic program as claimed.
\end{proof}


\subsection{Bell inequalities}\label{sec:two_player}

As discussed in the review article~\cite{Brunner14} Bell inequalities and two-player games are the same, and in the following we will take the game perspective. In particular, we are interested in a two-player game such that its classical value is related to the bounded norm of the condenser as in Proposition~\ref{prop:cond-to-norm}, and its entangled value relates to the completely bounded norm of the condenser as in Theorem~\ref{thm:cond-to-norm-quantum}. The game is as follows. There are the two players Alice and Bob, and a referee. The bipartite graph $G=(N,M,V,D)$ as defined by the condenser $\cond=\{f_{s}\}_{s\in D}$ via the neighbor function in~\eqref{eq:neighbor_graph} is known to all parties. First, the referee samples $\gamma=2^{n-k}$ elements $x_1,\dots,x_\gamma$ out of $N$ uniformly at random (with replacement), and likewise $\gamma'=2^{m-k'}$ independent and uniformly chosen random entries $y_1,\dots,y_{\gamma'}$ out of $M$, as well as a value of the seed $s$ according to the uniform distribution. We collect these random indices into vectors $\vec{x} \in N^\gamma$ and $\vec{y} \in M^{\gamma'}$; the questions for Alice and Bob, respectively. Alice and Bob then provide indices $1\leq \alpha \leq \gamma$ and $1\leq \beta \leq \gamma'$ of these vectors as answers. They win if
\begin{align}
\Gamma(x_\alpha,s) = y_\beta\ .
\end{align}
We use the notation $(G;2^{k},2^{k'})$ for this game. A classical strategy amounts to a pair of deterministic mappings $f: N^\gamma \to \{1,\dots,\gamma\}$ and $g:M^{\gamma'} \to \{1,\dots,\gamma'\}$ (independent of the value of $s$), and the classical value of the game is\footnote{Instead of deterministic functions $f$ and $g$, we could also allow for shared randomness, which does not increase the value of the game.}
\begin{align}
\omega\left(G;2^{k},2^{k'}\right):=\sup\Big\{ \E_s \E_{\vec{x},\vec{y}} \, \Gamma\big[\vec{x}(f(\vec{x})),s;\vec{y}(g(\vec{y}))\big]\Big\}\ ,
\end{align}
where the supremum is over all classical strategies, and
\begin{align}
\Gamma[x,s;y]\eqdef
\begin{cases}
1 & \mbox{for}\;\Gamma(x,s) = y\\
0 & \mbox{otherwise.}
\end{cases}
\end{align}
If Alice and Bob are allowed to share an entangled state of local dimension $Q$, they are not restricted to classical deterministic strategies. Instead, each player has a set of positive operator valued measures (POVMs), indexed by his questions, acting on his share of the entangled state. That is, Alice and Bob each have a set of positive operators in $\matr_Q$ labeled by the questions $\vec{x}$ and $\vec{y}$ and possible outcomes $\alpha$ and $\beta$,
\begin{align}
\hat{p}(\alpha;\vec{x})\geq0\quad&\mathrm{with}\quad\sum_\alpha \hat{p}(\alpha;\vec{x})\leq\1\label{eq:quantum_strategy1}\\
\hat{q}(\beta;\vec{y})\geq0\quad&\mathrm{with}\quad\sum_\beta \hat{q}(\beta;\vec{y})\leq\1\label{eq:quantum_strategy2}
\end{align}
for all questions $\vec{x}$ and $\vec{y}$ (again independent of the value of $s$). Note that we allow for incomplete POVMs (they do not sum to the identity), since we can always include a dummy answer into the game, associated to the missing normalization, thereby not changing the value of the game. For fixed $Q\in\mathbb{N}$ we define
\begin{align}
\omega^*_Q\left(G;2^{k},2^{k'}\right)\eqdef\sup \left\{ \E_s \E_{\vec{x},\vec{y}} \, \sum_{\alpha,\beta}\,\Gamma\big[\vec{x}(\alpha)),s;\vec{y}(\beta)\big]\cdot\bra{\psi} \hat{p}(\alpha;\vec{x}) \otimes \hat{q}(\beta;\vec{y})\ket{\psi}\,\right\}\ ,
\end{align}
where the supremum is over all bipartite pure state vectors $\ket{\psi}\in\C^Q \otimes\C^Q$, and all corresponding quantum strategies as in~\eqref{eq:quantum_strategy1}--\eqref{eq:quantum_strategy2}. The supremum over all bipartite pure state vectors just gives the operator norm and hence we have
\begin{align}
\omega^*_Q\left(G;2^{k},2^{k'}\right)=\sup\left\{\left\|\E_s \E_{\vec{x},\vec{y}}\,\sum_{\alpha,\beta}\,\Gamma\big[\vec{x}(\alpha)),s;\vec{y}(\beta)\big]\, \hat{p}(\alpha;\vec{x}) \otimes \hat{q}(\beta;\vec{y})\right\|_{(\infty;\infty)}\right\}\ .
\end{align} 
Finally, the entangled value of the game is given by
\begin{align}
\omega^*\left(G;2^{k},2^{k'}\right)\eqdef\sup_{Q\in\mathbb{N}}\omega^*_Q\left(G;2^{k},2^{k'}\right)\ .
\end{align}

If we ask condensers only to be quantum-proof against $Q$-dimensional quantum side information the relevant quantity to bound becomes (as in Theorem~\ref{thm:cond-to-norm-quantum})
\begin{align}
\|[\cond]\|_{(\capcap\{2^{k}\mathds{L}^\infty_N,\mathds{L}^1_M\})_Q\to(\sigsig\{2^{k'}\mathds{L}^\infty_M,\mathds{L}^1_M\})_Q}\eqdef\sup_{\stackrel{P\in\matr_{Q}(N)}{\|P\|_{(\capcap\{2^{k}\mathds{L}^\infty,\mathds{L}^1\})_Q\leq1}}}\|[\cond](P_N)\|_{(\sigsig\{2^{k'}\mathds{L}^\infty_M,\mathds{L}^1_M\})_Q}\ .
\end{align}
For the proof of Theorem~\ref{thm:condensergame} we will show that for any dimension $Q\in\mathbb{N}$,
\begin{align}\label{eq:condenser_basicthm}
\omega^{*}_Q(G;2^{k},2^{k'})\leq2^{-k'}\cdot\|[\cond]\|_{(\capcap\{2^{k}\mathds{L}^\infty_N,\mathds{L}^1_M\})_Q\to(\sigsig\{2^{k'}\mathds{L}^\infty_M,\mathds{L}^1_M\})_Q}\leq c\cdot\omega^{*}_Q(G;2^{k},2^{k'})\ ,
\end{align}
and with this prove all the claims at once (by letting $Q=1$ for~\eqref{eq:condensergame_b}, $Q\to\infty$ for~\eqref{eq:condensergame_cb}, or leaving $Q$ free). For the second inequality in~\eqref{eq:condenser_basicthm} we start from elements $p\in\matr_Q(\mathbb{C}^N)$ and $q\in \matr_Q(\mathbb{C}^M)$, satisfying
\begin{align}
\sum_x p(x) \leq \frac{2^{k}}{4}\cdot\1,\;p(x)\leq\1\quad\text{and}\quad\sum_y q(y)\leq\frac{2^{k'}}{4}\cdot\1,\;q(y)\leq\1\ ,
\end{align}
and construct POVMs $\hat{p}(\alpha;\vec{x})$ and $\hat{q}(\beta;\vec{y})$ associated to the questions and answers of the game such that
\begin{align}
\frac{1}{64}\cdot\frac{4}{2^{k}}\cdot\frac{4}{2^{k'}}\left\|\sum_{x,y} \,\E_s\,\Gamma[x,s;y]\,p(x) \otimes q(y)\right\|_{(\infty;\infty)}\leq\left\|\E_s \E_{\vec{x},\vec{y}} \sum_{\alpha, \beta}\Gamma\big[\vec{x}(\alpha),s;\vec{y}(\beta)\big]\,\hat{p}(\alpha;\vec{x})\otimes \hat{q}(\beta;\vec{y})\right\|_{(\infty;\infty)}\ .
\end{align}
Since the supremum of the left hand side over all such $q$ and $p$ is by operator space duality just $1/64$ times $\norm{\cond}_\cb$, and the statement follows by taking the expectation over the value of the seed $s$ on both sides. Conversely, the first inequality in~\eqref{eq:condenser_basicthm} is proven by constructing operators satisfying the $\capcap$-norm estimates out of POVM elements.

The POVMs $\hat{q}(\beta;\vec{y})$ and $\hat{p}(\alpha;\vec{x})$ are built in a similar manner. Let $\vec{x}$ denote again the vector with entries corresponding to the choice of $\gamma$ independent random indices in $N$. For $0\leq \alpha \leq \gamma$, $1 \leq l \leq \alpha$ or \(0 \leq \beta \leq \gamma'\), \(1 \leq l' \leq \beta\) we define for $p\in\matr_Q\left(\mathbb{C}^N\right)$ and $q\in\matr_Q\left(\mathbb{C}^M\right)$,
\begin{align}\label{eq:rt1}
r_{\vec{x},l,\alpha}\eqdef\prod_{l\leq k \leq \alpha} \big(\1 - p[\vec{x}(k)]\big)^{1/2}\ ,\quad t_{\vec{y},l',\beta}\eqdef\prod_{l'\leq k \leq \beta} \big(\1 - q[\vec{y}(k)]\big)^{1/2}
\end{align}
as well as
\begin{align}\label{eq:rt2}
r_{\vec{x},\alpha+1,\alpha}  \eqdef\1 \;,\;\;r_{\vec{x},\alpha} \eqdef r_{\vec{x},1,\alpha}\quad\mathrm{and}\quad t_{\vec{y},\beta+1,\beta}\eqdef \1 \;,\;\;t_{\vec{y},\beta} \eqdef t_{\vec{y},1,\beta}\ .
\end{align}
The POVMs are then constructed as
\begin{align}
\hat{p}(\alpha;\vec{x})\eqdef r_{\vec{x},\alpha-1}^*\,p[\vec{x}(\alpha)]\,r_{\vec{x},\alpha-1}\quad\mathrm{and}\quad\hat{q}(\beta;\vec{y})\eqdef t_{\vec{y},\beta-1}^*\,q[\vec{y}(\beta)]\,t_{\vec{y},\beta-1}\ .
\end{align}
The following lemma asserts that these are valid POVMs and also provides some estimates for remaining terms, to be used later on. It is a special case of~\cite[Lemma 6.7]{Junge:2006wt}, but for the convenience of the reader we restate it, and provide an adapted proof in Appendix~\ref{app:missing}.

\begin{restatable}{lemma}{jungetechnical}\label{lem:Jungelemma}
Let $p\in\matr_Q\left(\mathbb{C}^N\right)$, and $\alpha$, $\gamma$, $r_{\vec{x},\alpha}$ as in~\eqref{eq:rt1}--\eqref{eq:rt2}. Then, we have
\begin{align}
&\sum_{\alpha=1}^{\gamma} r_{\vec{x},\alpha-1}^* p[\vec{x}(\alpha)] r_{\vec{x},\alpha-1} \leq \1\label{eq:junge1}\\
&\sum_{\alpha=1}^{\gamma} \E_{x_1,\dots,x_{\alpha-1}} \left(\1 - r_{\vec{x},\alpha-1}\right) \leq \frac{\gamma}{8}\cdot\1\label{eq:junge2}\\
&\sum_{\alpha=1}^{\gamma} \E_{x_1,\dots,x_{\alpha-1}} [\1 - r_{\vec{x},\alpha-1}^*] [\1 - r_{\vec{x},\alpha-1}] \leq \frac{\gamma}{4}\cdot\1\ .\label{eq:junge3}
\end{align}
Similar estimates hold for $q\in\matr_Q\left(\mathbb{C}^M\right)$ with the replacements $\alpha\mapsto\beta$, $\gamma\mapsto\gamma'$, and $r_{\vec{x},\alpha}\mapsto t_{\vec{y},\beta}$ as defined in~\eqref{eq:rt1}--\eqref{eq:rt2}.
\end{restatable}
 
\begin{proof}[Proof of Theorem~\ref{thm:condensergame}]
We start with the direction game $\implies$ condenser. The main idea of the proof is similar to the proof of~\cite[Proposition 6.9]{Junge:2006wt}. According to the rule
\begin{align}
t^* q t= q + (1-t^*) 1 + q(1-t) - (1-t^*) q (1-t) \leq q + (1-t^*) 1 + q(1-t) + (1-t^*) q (1-t)\ ,
\end{align}
we split the following sum into $16$ terms,
\begin{align}\label{eq:condensergameexpandterms}
&\sum_{\alpha,\beta}\Gamma[\vec{x}(\alpha),s;\vec{y}(\beta)] \,  p[\vec{x}(\alpha)]  \otimes  q[\vec{y}(\beta)]  \notag\\
&\leq \sum_{\alpha,\beta}   \Gamma[\vec{x}(\alpha),s;\vec{y}(\beta)] \, r_{\vec{x},\alpha-1}^* p[\vec{x}(\alpha)] r_{\vec{x},\alpha-1} \otimes  t_{\vec{y},\beta-1}^* q[\vec{y}(\beta)] t_{\vec{y},\beta-1} \notag\\
&\quad+ \sum_{\alpha,\beta}   \Gamma[\vec{x}(\alpha),s;\vec{y}(\beta)] \, [\1 - r_{\vec{x},\alpha-1}^*] p[\vec{x}(\alpha)]  \otimes  [\1 - t_{\vec{y},\beta-1}^*] q[\vec{y}(\beta)] \notag\\
&\quad+\sum_{\alpha,\beta}   \Gamma[\vec{x}(\alpha),s;\vec{y}(\beta)] \, p[\vec{x}(\alpha)] [\1 - r_{\vec{x},\alpha-1}] \otimes   q[\vec{y}(\beta)] [\1 - t_{\vec{y},\beta-1}] \notag\\
&\quad+ \dots \text{ 12 other terms } \notag\\
&\quad+ \sum_{\alpha,\beta}   \Gamma[\vec{x}(\alpha),s;\vec{y}(\beta)] \,[\1 - r_{\vec{x},\alpha-1}^*] p[\vec{x}(\alpha)] [\1 - r_{\vec{x},\alpha-1}] \otimes  [\1 - t_{\vec{y},\beta-1}^*]  q[\vec{y}(\beta)] [\1 - t_{\vec{y},\beta-1}]\ .
\end{align}
We now take the expectation over the random choices $\vec{x}$ and $\vec{y}$ and the seed value \(s\). The left hand side reduces to
\begin{align}
\E_s \E_{\vec{x},\vec{y}} \sum_{\alpha,\beta}\Gamma[\vec{x}(\alpha),s;\vec{y}(\beta)] \,  p[\vec{x}(\alpha)]\otimes q[\vec{y}(\beta)]&= \E_s \E_{x_1,\dots,x_\gamma,y_1,\dots,y_{\gamma'}} \sum_{\alpha,\beta} \Gamma[x_\alpha,s;y_\beta]\,p[x_\alpha]\otimes q[y_\beta]\\
&= \gamma \,\gamma' \, \E_s \E_x \E_y \Gamma[x,s;y] p(x)\otimes q(y)\ .
\end{align}
The first term on the right hand side of \eqref{eq:condensergameexpandterms} is the Bell (or game) operator which norm is equal to the value of the strategy given by the POVMs $\hat{q}(\beta;\vec{y})$ and $\hat{p}(\alpha;\vec{x})$. Let us examine the remaining terms. Expressions such as the second or third one are evaluated to
\begin{align}
&\sum_{\alpha,\beta}\E_s \E_{x_\alpha} \E_{y_\beta}\Big\{\Gamma[\vec{x}(\alpha),s;\vec{y}(\beta)] \, \left(\E_{x_1,\dots,x_{\alpha-1}}[\1 - r_{\vec{x},\alpha-1}^*]\right) p[\vec{x}(\alpha)]  \otimes  \left(\E_{y_1,\dots,y_{\beta-1}}[\1 - t_{\vec{y},\beta-1}^*]\right) q[\vec{y}(\beta)]\Big\} \notag\\
&\leq \frac{\gamma \, \gamma'}{8 \cdot 8} \norm{\E_s \E_x \E_y \Gamma[x,s;y] p(x) \otimes q(y)}_{(\infty;\infty)}\cdot\1\ ,
\end{align}
where we applied Lemma~\ref{lem:Jungelemma}, \eqref{eq:junge2}, to the operators in the brackets. Note that we had to use that the estimate are independent of the seed value. Terms involving $[\1 - r_{\dots}]$ or its starred version on both sides of $p[\dots]$ (and similarly those involving $[\1 - t_{\dots}]$ on both sides of $q[\dots]$) are estimated as follows
\begin{align}
&\sum_{\alpha,\beta}\E_s \E_{x_\alpha} \E_{y_\beta}\notag\\
&\Big\{\Gamma[\vec{x}(\alpha),s;\vec{y}(\beta)] \,\E_{x_1,\dots,x_{\alpha-1}} \left\{[\1 - r_{\vec{x},\alpha-1}^*] p[\vec{x}(\alpha)] [\1 -r_{\vec{x},\alpha-1}]\right\}\otimes\E_{y_1,\dots,y_{\beta-1}} \left\{[\1 - t_{\vec{y},\beta-1}^*]  q[\vec{y}(\beta)] [\1 - t_{\vec{y},\beta-1}]\right\}\Big\} \notag\\
&\leq \norm{\E_s \E_x \E_y \Gamma[x,s;y] p(x) \otimes q(s,y)}_{(\infty;\infty)}\cdot\left\|\sum_\alpha \E_{x_1,\dots,x_{\alpha-1}} [\1 - r_{\vec{x},\alpha-1}^*] [\1 - r_{\vec{x},\alpha-1}]\right\|_{(\infty;\infty)}\notag\\
&\quad\,\,\cdot\max_s \,\left\|\sum_\beta\E_{y_1,\dots,y_{\beta-1}} [\1 - t_{s,\vec{y},\beta-1}^*] [\1 - t_{\vec{y},\beta-1}]\right\|_{(\infty;\infty)}\cdot\1\\
&\leq  \frac{\gamma \, \gamma'}{4 \cdot 4} \norm{\E_s \E_x \E_y \Gamma[x,s;y] p(x) \otimes q(y)}_{(\infty;\infty)}\cdot\1\ .
 \end{align}
Here we used that for completely positive maps $\phi$ and positive operators $a$ it holds that $\phi(a) \leq \norm{a} \,\phi(\1)$, which we applied to the cp maps given by the Kraus operators \(\1 - r_{\vec{x},\alpha-1}\) and \(\1 - t_{\vec{y},\beta-1}\). The estimate then followed by applying Lemma~\ref{lem:Jungelemma}, \eqref{eq:junge3}, and again using that the estimate are independent of the seed value. Estimating all cross terms according to the two strategies leads to the final estimate
\begin{align}
\frac{1}{4}\cdot2^{-k-k'}\left\|\sum_{x,y} \E_s  \Gamma[x,s;y] p(x) \otimes q(y)\right\|_{(\infty;\infty)}\,&=\, \gamma \,\gamma' \, \left\|\E_s \E_x \E_y \Gamma[x,s;y] p(x) \otimes q(y)\right\|_{(\infty;\infty)}\\
&\leq\left\|\sum_{\alpha,\beta} \E_s \E_{\vec{x},\vec{y}}  \Gamma[\vec{x}(\alpha),s;\vec{y}(\beta)] \,\hat{p}(\alpha;\vec{x}) \otimes  \hat{q}(\beta;\vec{y})\right\|_{(\infty;\infty)}\ .
\end{align}
 
For the converse part, let $\hat{q}(\beta;\vec{y})$ and $\hat{p}(\alpha;\vec{x})$ be arbitrary quantum strategies. We perform the following transformation
\begin{align}
\E_s \E_{\vec{x},\vec{y}} \sum_{\alpha, \beta}\Gamma[\vec{x}(\alpha),s;\vec{j}(\beta)] \hat{p}(\alpha;\vec{x}) \otimes  \hat{q}(\beta;\vec{y})= \sum_{x'\in N, y' \in M} \,\E_s\, \Gamma[x,s;y] \sum_\alpha \E_{\vec{x}} \left[ \delta_{\vec{x}(\alpha) = x'} \hat{p}(\alpha;\vec{x})\right] \otimes  \sum_\beta \E_{\vec{y}} \left[ \delta_{\vec{y}(\beta) = y'} \hat{q}(\beta;\vec{y})\right]\ .
\end{align}
We now show that the collection of operators 
\begin{align}\label{eq:definecondenseropsthroughbellops}
p(x')= \sum_\alpha \E_{\vec{x}} \left[ \delta_{\vec{x}(\alpha) = x'} \hat{p}(\alpha;\vec{x})\right]\quad\mathrm{and}\quad q(y')= \sum_\beta \E_{\vec{y}} \left[ \delta_{\vec{y}(\beta) = y'} \hat{q}(\beta;\vec{y})\right]
\end{align}
satisfy the $\capcap$-norm estimates
\begin{align}\label{eq:gametocondenserestimates}
\sum_{x'} p(x') \leq \1\quad\mathrm{and}\quad p(x') \leq2^{-k}\1\quad\mathrm{plus}\quad\sum_{y'} q(y') \leq \1\quad\mathrm{and}\quad q(y') \leq2^{-k'}\1\ .
\end{align}
Hence, we have 
\begin{align}
&\left\|\E_s \E_{\vec{x},\vec{y}} \sum_{\alpha, \beta}  \Gamma[\vec{x}(\alpha),s;\vec{j}(\beta)] \hat{p}(\alpha;\vec{x}) \otimes  \hat{q}(\beta;\vec{y})\right\|_{(\infty;\infty)}\notag\\
&\leq \sup\left\{\left\| \sum_{x,y} \,\E_s  \Gamma[x,s;y] p(x) \otimes q(y)\right\|_{(\infty;\infty)}\;:\;\norm{p}_{\capcap\{2^{k}\ell^\infty,\ell^1\}} \leq 1\,,\; \norm{q}_{\capcap\{2^{k'}\ell^\infty,\ell^1\}} \leq 1 \right\}\ ,
\end{align}
and again operator space duality provides the last argument. 
 
To show the first set of estimates in~\eqref{eq:gametocondenserestimates}, note
\begin{align}
\sum_{x'} \sum_\alpha \E_{\vec{x}} \left[ \delta_{\vec{x}(\alpha) = x'} \hat{p}(\alpha;\vec{x})\right] = \sum_\alpha \E_{\vec{x}} \left[  \sum_{x'} \delta_{\vec{x}(\alpha) = x'} \hat{p}(\alpha;\vec{x})\right] =  \sum_\alpha \, \hat{p}(\alpha;\vec{x}) \leq \1 \,,
\end{align}
and similarly $\sum_{y',\beta} \E_{\vec{y}} \left[ \delta_{\vec{y}(\beta) = y'} \hat{q}(\beta;\vec{y})\right] \leq \1$. Moreover, we have $\hat{p}(\alpha;\vec{x}) \leq \1$ and hence
\begin{align}
\sum_\alpha \E_{\vec{x}} \left[ \delta_{\vec{x}(\alpha) = x'} \hat{p}(\alpha;\vec{x})\right] \leq \sum_\alpha \P[x_\alpha = x'] \,\leq \frac{\gamma}{N} \1 =2^{-k}\1\ ,
\end{align}
and again similarly $\E_{\vec{y}} \left[ \delta_{\vec{y}(\beta) = y'} \hat{q}(\beta;\vec{y})\right]\leq2^{-k'}$.
\end{proof}


\section*{Acknowledgments}

We acknowledge discussions with Matthias Christandl, Fabian Furrer, Patrick Hayden, Christopher Portmann, Renato Renner, Oleg Szehr, Marco Tomamichel, Thomas Vidick, Stephanie Wehner, Reinhard Werner, and Andreas Winter. MB acknowledges funding provided by the Institute for Quantum Information and Matter, an NSF Physics Frontiers Center (NFS Grant PHY-1125565) with support of the Gordon and Betty Moore Foundation (GBMF-12500028). Additional funding support was provided by the ARO grant for Research on Quantum Algorithms at the IQIM (W911NF-12-1-0521). Most of this work was done while OF was affiliated with ETH Zurich, supported by the the European Research Council grant No. 258932. Additional funding was provided by the EU under the project ``Randomness and Quantum Entanglement'' (RAQUEL). VBS was supported by an ETH postdoctoral fellowship. Part of this work was done while OF and VBS were visiting the Institute for Quantum Information and Matter at Caltech and we would like to thank John Preskill and Thomas Vidick for their hospitality.


\appendix

\section{Missing proofs}\label{app:missing}

\begin{theorem}[Grothendieck's inequality]\label{thm:grothendieck}
For any real matrix $\{A_{ij}\}$, we have
\begin{align}
\max\left\{ \sum_{i,j} A_{ij} \vec{a}_{i} \cdot \vec{b}_j : \| \vec{a}_{i} \|_2 \leq 1, \| \vec{b}_j \|_2 \leq 1\right\}\leq K_G \cdot \max \left\{ \sum_{i,j} A_{ij} a_i b_j : a_i, b_j \in \bR, |a_i| \leq 1, |b_j| \leq 1 \right\}\ .
\end{align}
\end{theorem}

\begin{proposition}\label{prop:duality-one-infty}
On $\matr_{QN}$ the norm $\|\cdot\|_{(1;\infty)}$ is dual to the norm $\|\cdot\|_{(\infty;1)}$.
\end{proposition}

\begin{proof}
For $A \in \matr_{QN}$ we calculate
\begin{align}
\| A \|_{(1;\infty)^{*}}&=\sup\left\{\tr\left[ B^{*} A \right] \ : \ \| B \|_{(1;\infty)} \leq 1\right\}\\
&= \sup\left\{ \tr\left[ B^{*} A \right] \ : \  B = (D_1 \otimes \1_N) C (D_2 \otimes \1_N), \| D_1 \|_2, \| D_2 \|_2 \leq 1, \| C \|_{(\infty;\infty)} \leq 1\right\}\\
&= \sup\left\{\tr\left[C^{*} (D_1^{*} \otimes \1_N) A (D_2^{*} \otimes \1_N)  \right],\| D_1 \|_2, \|D_2 \|_2 \leq 1, \| C \|_{(\infty;\infty)} \leq 1\right\}\\
&= \sup\left\{\| (D_2^{*} \otimes \1_N) A (D_1^{*} \otimes \1_N) \|_{(1;1)} , \| D_1 \|_2, \|D_2 \|_2 \leq 1 \right\}\\
&= \| A \|_{(\infty;1)}\ ,
\end{align}
where we have used that $\|\cdot\|_{(\infty;\infty)}$ is dual to $\|\cdot\|_{(1;1)}$.
\end{proof}

\begin{proposition}\label{prop:min-ent-norm}
For $\rho_{QN}\in\cS(QN)$ we have $\|\rho_{QN}\|_{(1;\infty)} = 2^{-H_{\min}(N|Q)_{\rho}}$. 
\end{proposition}

\begin{proof}
This is basically proven in~\cite{DJKR06} but we reproduce the argument here for convenience. Let $\sigma_{Q}\in\cS(Q)\subset\matr_{Q}$ be such that $2^{-H_{\min}(N|Q)_{\rho}} = \left\|\left(\sigma^{-1/2}_{Q}\otimes\1_{N}\right)\rho_{QN}\left(\sigma^{-1/2}_{Q}\otimes\1_{N}\right)\right\|_{(\infty;\infty)}$. Then, we have
\begin{align}
\|\rho_{QN}\|_{(1;\infty)}&\leq\left\|\sigma^{1/2}_{Q}\right\|_2 \left\|\left(\sigma^{-1/2}_{Q}\otimes\1_{N}\right)\rho_{QN}\left(\sigma^{-1/2}_{Q}\otimes\1_{N}\right)\right\|_{(\infty;\infty)}\left\|\sigma^{1/2}_{Q}\right\|_2 = 2^{-H_{\min}(N|Q)_{\rho}}\ .
\end{align}
For the other direction we will show that for $\rho_{QN}\geq0$,
\begin{align}\label{eq:final_duality}
\|\rho_{QN}\|_{(1;\infty)}=\inf \{ \| A \|_2 \|\omega_{QN}\|_{(\infty;\infty)} \| A \|_2 : \rho_{QN}=(A \otimes \1_{N}) \omega_{QN} (A \otimes \1_{N}); A\in \matr_Q, A \geq 0 \}\ ,
\end{align}
from which the claim follows.

We prove the similar statement for the dual norm $\|\cdot\|_{(\infty;1)}$ first. For that let $X \geq 0$ and using H\"{o}lder's inequality we get
\begin{align}
\| A X B^{*} \|_1 = \left\| A \sqrt{X} \sqrt{X} B^{*} \right\|_1 &\leq \left\| A \sqrt{X} \right\|_2 \left\| \sqrt{X} B^{*} \right\|_2\leq \max\big\{\| A X A^{*} \|_1, \| B X B^{*} \|_1\big\}\ .
\end{align}
By performing a polar decomposition of $A$ and using the unitary invariance of the norm $\|\cdot\|_1$, we get that 
\begin{align}
\| X \|_{(\infty;1)} = \sup \{ \| A X A \|_1 : A\in \matr_Q, A \geq 0, \|A\|_2 \leq 1 \}\ .
\end{align}
By the same arguments which show that the norm $\|\cdot\|_{(\infty;1)}$ is dual to the norm $\|\cdot\|_{(1;\infty)}$ (Proposition~\ref{prop:duality-one-infty}), this then implies~\eqref{eq:final_duality}.
\end{proof}

\begin{lemma}\label{lem:calc-vect-to-matr}
Let $u : \matr_N \to \matr_Q$, $z\in\matr_{QN}$, as well as $a = \sum_{ij} a_{i,j} \ket{i} \otimes \ket{j}, b = \sum_{i,j} b_{ij} \ket{i} \otimes \ket{j}$ vectors in $\C^{Q^2}$. Then, we have
\begin{align}
\bra{b} (\1_{Q} \otimes u)(z) \ket{a} = \tr\left[  J(u) (A^T \otimes \1_{N}) z^T (\bar{B} \otimes \1_{N}) \right] \ ,
\end{align}
where $J(u) = \left(\sum_{k,k'} u(\ket{k} \bra{k'}) \otimes \ket{k} \bra{k'} \right)$, $A = \sum_{i,j} a_{ij} \ket{i} \bra{j}$, and $B = \sum_{i,j} b_{ij} \ket{i} \bra{j}$.
\end{lemma}

\begin{proof}
We write
\begin{align}
\bra{b} (\1_{Q} \otimes u)(z) \ket{a}
&= \sum_{i,i',j,j'} \tr[ (\1_Q \otimes u)(z) a_{ij} \bar{b}_{i'j'} \ket{i} \bra{i'} \otimes \ket{j} \bra{j'} ] \\
&= \sum_{i,i',j,j', k,k'} z_{i'ikk'} a_{ij} \bar{b}_{i'j'} \tr[ \ket{i} \bra{i'} \otimes u( \ket{k} \bra{k'} ) \ket{i'} \bra{i} \otimes \ket{j} \bra{j'}  ] \\
&= \sum_{i,i',j,j', k,k'} z_{i'ikk'} a_{ij} \bar{b}_{i'j'} \tr[ x \ket{j} \bra{j'} \otimes \ket{k'} \bra{k}  ] \\
&= \sum_{i,i', k,k'} z_{i'ikk'}  \tr\left[ J(u) \left(\sum_{j} a_{ij} \ket{j} \right) \left(\sum_{j'} \bar{b}_{i'j'} \bra{j'}\right) \otimes \ket{k'} \bra{k} \right] \\
&= \sum_{i,i', k,k'} z_{i'ikk'} \tr\left[ J(u) A^T \ket{i} \bra{i'} \bar{B} \otimes \ket{k'} \bra{k} \right] \\
&=  \tr\left[ (\bar{B} \otimes \1_{N}) J(u) (A^T \otimes \1_{N}) z^{T} \right] \ .
\end{align}
\end{proof}

\jungetechnical*

\begin{proof}
For the first property, note that
\begin{align}
\1 - \sum_\alpha r_{\vec{x},\alpha-1}^* p[\vec{x}(\alpha)] r_{\vec{x},\alpha-1} &= \1 - p[x_1] - \sum_{\alpha=2} r_{\vec{x},\alpha-1}^* p[\vec{x}(\alpha)] r_{\vec{x},\alpha-1}\\
&= (\1 - p[x_1])^{1/2} \left( \1 - \sum_{\alpha=2} r_{\vec{x},2,\alpha-1}^* p[\vec{x}(\alpha)] r_{\vec{x},2,\alpha-1} \right) (\1 - p[x_1])^{1/2}\\
&=\dots = \, r_{\vec{x},\gamma}^* r_{\vec{x},\gamma}\\
&\geq0\ .
\end{align}
For the second property, it follows from $1 - \sqrt{1-t} \leq t$, $t \in [0,1]$ that $\1 - (\1 - p(k))^{1/2} \leq p(k)$ and hence also $(1- \frac{1}{4 \,\gamma})\1 \leq \frac{1}{N}\sum_x (\1 - p(x))^{1/2}$. We then have 
\begin{align}
\sum^\gamma_\alpha\E_{x_1,\dots,x_{\alpha-1}} \left(\1 - r_{\vec{x},\alpha-1}\right)=\sum^\gamma_{\alpha=1} \left(\1 - \E_{x_1}\,\E_{x_2}\, \dots \,\E_{x_{\alpha-1}} r_{\vec{x},\alpha-1}\right)&= \sum^\gamma_{\alpha=1} \left(\1 - \left(\frac{1}{N}\sum_x (\1 - p(x))^{1/2}\right)^{\alpha-1} \right)\\
&\leq \sum^\gamma_{\alpha=1} \left(1 - \left(1 - \frac{1}{4 \,\gamma}\right)^{\alpha-1} \right) \,\1\\
& \leq \gamma (4 (1 - \frac{1}{4\gamma})^\gamma -3) \,\1\\
&\leq \gamma (4 \,e^{-1/4} - 3) \,\1\\
&\leq \frac{\gamma}{8} \,\1\ .
\end{align}
The last assertion follows by
\begin{align}
&\sum_\alpha \E_{x_1,\dots,x_{\alpha-1}} [\1 - r_{\vec{x},\alpha-1}^*] [\1 - r_{\vec{x},\alpha-1}]\notag\\ 
&=\1 - \sum_\alpha \E_{x_1,\dots,x_{\alpha-1}} [r_{\vec{x},\alpha-1}^*] - \sum_\alpha \E_{x_1,\dots,x_{\alpha-1}} [r_{\vec{x},\alpha-1}] + \sum_\alpha \E_{x_1,\dots,x_{\alpha-1}} [r_{\vec{x},\alpha-1}^*\, r_{\vec{x},\alpha-1}]\\
&\leq 2 \left( \1 - \left(\frac{1}{N}\sum_x (\1 - p(x))^{1/2}\right)^{\alpha-1} \right)\ ,
\end{align}
since $r_{\vec{x},\alpha-1}^*\, r_{\vec{x},\alpha-1} \leq \1$. The proof for terms involving \(t\)'s and \( q\)'s is identical. 
\end{proof}


\section{Haagerup tensor product and intersection norms}\label{app:intersection}

We denote the Haagerup tensor product of operator spaces by $\otimes_h$ (also for the corresponding normed spaces tensor product). We refer to~\cite[Chapter 5]{Pis03} for details, the calculations we perform here are based on a few simple properties. For operator spaces $\mathds{E}=\left(E,\|\cdot\|_{\mathds{E}_Q}\right)$ and $\mathds{F}=\left(F,\|\cdot\|_{\mathds{F}_Q}\right)$ of dimension $N$, we have
\begin{align}\label{eq:def-haagerup}
\| x \|_{\mathds{E}\otimes_{h}\mathds{F}} = \inf \left\{ \left\| \sum_{i} \ket{0} \bra{i} \otimes a_i\right\|_{\mathds{E}_N} \cdot \left\| \sum_{i} \ket{i} \bra{0} \otimes b_i \right\|_{\mathds{F}_N} : x = \sum_{i=1}^N a_i \otimes b_i \right\}\ .
\end{align}
Recall that we previously mentioned the column and row operator spaces $\mathds{C}_N=\left(C_N,\|\cdot\|_{(\mathds{C}_N)_Q}\right)$ and $\mathds{R}_N=\left(R_N,\|\cdot\|_{(\mathds{R}_N)_Q}\right)$. They are simply defined by embedding a vector $\bC^N$ as a column or row of a matrix in $\matr_N$. So we have for $x = \sum_{ijk} x_{ijk} \ket{i} \bra{j} \otimes \ket{k} \in \matr_Q(\mathbb{C}^N)$, 
\begin{align}
\| x \|_{(\mathds{C}_N)_Q}= \left\| \sum_{ijk} x_{ijk} \ket{i} \bra{j} \otimes \ket{k} \bra{0} \right\|_{(\infty;\infty)}\quad\mathrm{and}\quad\| x \|_{(\mathds{R}_N)_Q}= \left\| \sum_{ijk} x_{ijk} \ket{i} \bra{j} \otimes \ket{0} \bra{k} \right\|_{(\infty;\infty)}\ .
\end{align}
Using the $\min$ tensor product notation (see~\cite[Chapter 2]{Pis03} for a definition), the $\mathds{E}_N$- and $\mathds{F}_N$-norm in~\eqref{eq:def-haagerup} can be written as $\mathds{L}^\infty_N\omin\mathds{E}$ and $\mathds{L}^\infty_N\omin\mathds{F}$, respectively. As only the first row (column) of the matrix is used, we can write the norm as $\mathds{R}_N\omin\mathds{E}$ ($\mathds{C}_N\omin\mathds{E}$):
\begin{align}\label{eq:def-haagerup-2}
\| x \|_{\mathds{E} \otimes_{h} \mathds{F}} = \inf \left\{ \left\| \sum_{i} \bra{i} \otimes a_i \right\|_{\mathds{R}_N\omin\mathds{E}} \cdot \left\| \sum_{i} \ket{i} \otimes b_i \right\|_{\mathds{C}_N\omin\mathds{F}} : x = \sum_{i=1}^N a_i \otimes b_i \right\}\ .
\end{align}
In order to prove Proposition~\ref{prop:normeqhaageruptp}, we will make use of the following complete isometries (operator space equalities):
\begin{align}
&\mathds{C}_N\otimes_h\mathds{E}=\mathds{C}_N\omin\mathds{E}\\
&\mathds{E}\otimes_h\mathds{R}_N=\mathds{E}\omin\mathds{R}_N\\ 
&\mathds{C}_N\omin\mathds{R}_N=\mathds{L}^\infty_N\\
& \mathds{L}_Q^{\infty}(\mathds{E}) = \mathds{C}_Q \otimes_h \mathds{E}\otimes_h\mathds{R}_Q\\
& \mathds{L}_Q^\infty\omin\cap\left\{\mathds{E},\mathds{F}\} = \cap\{\mathds{L}_Q^\infty\omin\mathds{E},\mathds{L}_Q^\infty\omin\mathds{F}\right\}\ .
\end{align}

\begin{proposition}\label{prop:normeqhaageruptp}
The operator space $\capcap\left\{2^k \mathds{L}^\infty_N, \mathds{L}^1_N\right\}$ is completely isomorphic to the operator space
\begin{align}
\cap\left\{\sqrt{2^k} \mathds{C}_N, \mathds{R}_N\right\} \otimes_h \cap\left\{\sqrt{2^k} \mathds{R}_N, \mathds{C}_N\right\}\ .
\end{align}
\end{proposition}

\begin{proof}
We use the abbreviations
\begin{align}
\mathds{G}:=\cap\left\{\sqrt{2^k}\mathds{C}_N,\mathds{R}_N\right\}\otimes_h\cap\left\{\sqrt{2^k}\mathds{R}_N,\mathds{C}_N\right\}\quad\mathrm{and}\quad\mathds{G}_Q:=\left(\cap\left\{\sqrt{2^k}\mathds{C}_N,\mathds{R}_N\right\}\otimes_h\cap\left\{\sqrt{2^k} \mathds{R}_N,\mathds{C}_N\right\}\right)_Q\ ,
\end{align}
with the operator space structure
\begin{align}
\text{$\mathds{L}^{\infty}_Q(\mathds{G}):=\left(\matr_Q(\matr_N),\|\cdot\|_{(\mathds{L}^{\infty}_Q(\mathds{G}))_{Q'}}\right)$ given by $\|\cdot\|_{(\mathds{L}^{\infty}_Q(\mathds{G}))_{Q'}}:=\|\cdot\|_{\mathds{G}_{QQ'}}$.}
\end{align}
The rules for Haagerup tensor products give
\begin{align}
\mathds{L}^{\infty}_Q(\mathds{G})&= \mathds{C}_Q \otimes_h \cap\left\{\sqrt{2^k} \mathds{C}_N, \mathds{R}_N\right\} \otimes_h \cap\left\{\sqrt{2^k} \mathds{R}_N, \mathds{C}_N\right\} \otimes_h \mathds{R}_Q \\
&= \mathds{C}_Q \omin \cap\left\{\sqrt{2^k} \mathds{C}_N, \mathds{R}_N\right\} \otimes_h \cap\left\{\sqrt{2^k} \mathds{R}_N, \mathds{C}_N\right\} \omin \mathds{R}_Q \\
&= \cap\left\{\sqrt{2^k} \mathds{C}_Q \omin \mathds{C}_N, \mathds{C}_Q \omin \mathds{R}_N \right\} \otimes_h \cap\left\{\sqrt{2^k} \mathds{R}_N \omin \mathds{R}_Q, \mathds{C}_N \omin \mathds{R}_Q\right\} \\
&= \cap\left\{\sqrt{2^k} \mathds{C}_{QN}, \mathds{C}_Q \omin \mathds{R}_N \right\} \otimes_h \cap\left\{\sqrt{2^k} \mathds{R}_{NQ} , \mathds{C}_N \omin \mathds{R}_Q\right\} \ .
\end{align}
Then, considering the definition in~\eqref{eq:def-haagerup}, we continue with
\begin{align}
\mathds{R}_{QN}\omin\cap\left\{\sqrt{2^k}\mathds{C}_{QN},\mathds{C}_Q\omin\mathds{R}_N\right\}&=\cap\left\{\sqrt{2^k}\mathds{L}^\infty_{QN},\mathds{R}_{QN}\omin\mathds{C}_Q\omin\mathds{R}_N\right\}\\
\mathds{C}_{QN}\omin\cap\left\{\sqrt{2^k}\mathds{R}_{NQ},\mathds{C}_N\omin\mathds{R}_Q\right\}&=\cap\left\{\sqrt{2^k}\mathds{L}^\infty_{QN},\mathds{C}_{QN}\omin\mathds{C}_N\otimes_h\mathds{R}_Q\right\}\ .
\end{align}
Now let $x \in \matr_{QN}$ with $x = \sum_{ijkl} x_{ijkl} \ket{i} \bra{j} \otimes \ket{k} \bra{l}$: our objective is to compute $\| x \|_{\mathds{G}_Q}$. Seeing $x$ as an element of $\bC^{QN} \otimes \bC^{QN}$, we find using the expression in~\eqref{eq:def-haagerup-2},
\begin{align}
\| x \|_{\mathds{G}_Q}=\inf\Bigg\{&\max\Bigg\{ \sqrt{2^k}\bigg\| \sum_{ikp} a_{ik}(p) \ket{i} \ket{k} \bra{p} \bigg\|_{(\infty;\infty)},\bigg\| \sum_{ikp} a_{ik}(p) \ket{i} \bra{k} \bra{p}\bigg\|_{\infty}\Bigg\}\\
\cdot&\max\Bigg\{ \sqrt{2^k}\bigg\| \sum_{jlp} b^*_{jl}(p) \ket{p} \bra{j} \bra{l}  \bigg\|_{(\infty;\infty)},\bigg\| \sum_{jlp} b^*_{jl}(p)\ket{p} \ket{j} \bra{l}\bigg\|_{\infty}\Bigg\}:x = \sum_{p=1}^{NQ} a(p) \otimes b(p)\Bigg\}\ ,
\end{align}
where we have used the notation $a(p) = \sum_{ik} a_{ik}(p) \ket{i} \ket{k}$ and $b(p) = \sum_{jl} b^*_{jl}(p) \bra{j} \bra{l}$. If we define $a = \sum_{ikp} a_{ik}(p)\ket{i} \ket{k} \bra{p}$ and $b = \sum_{jlp} b_{jl}(p) \ket{j} \ket{l} \bra{p}$, then $x = \sum_{p} a(p) \otimes b(p)$ becomes $x = a b^*$. We find exactly the expression in \eqref{eq:def-cap-cap}, which proves the assertion.
\end{proof}

We end this appendix with an identification of the dual operator space of $\capcap\left\{2^k \mathds{L}^\infty_N, \mathds{L}^1_N\right\}$. Since the Haagerup tensor product is self dual~\cite[Chapter 5]{Pis03}, Proposition~\ref{prop:normeqhaageruptp} is conveniently applied. First, note that by the discussion on intersection norms for operator spaces in the previous section, the dual operator space of $\cap\left\{\sqrt{2^k} \mathds{C}_N, \mathds{R}_N\right\}$ is equal to $\Sigma\left\{2^{-k/2} \mathds{R}_N, \mathds{C}_N\right\}$, since the operator space dual of the row operator space is the column space and vice versa.

\begin{corollary}
The operator space dual of $\capcap\left\{2^k \mathds{L}^\infty_N, \mathds{L}^1_N\right\}$ is the operator space
\begin{align}
\Sigma\left\{2^{-k/2} \mathds{R}_N, \mathds{C}_N \right\} \otimes_h \Sigma\left\{2^{-k/2} \mathds{R}_N, \mathds{C}_N \right\} = 2^{-k}\left(\Sigma\left\{\mathds{R}_N,\sqrt{2^k} \mathds{C}_N\right\} \otimes_h \Sigma\left\{\mathds{R}_N,\sqrt{2^k} \mathds{C}_N\right\}\right)\ .
\end{align}
\end{corollary}


\bibliographystyle{arxiv}
\bibliography{library}

\end{document}

%% file: volkherincludes.tex
\def\id{{\rm id}}                            
\def\Rl{{\mathbb R}}\def\C{{\mathbb C}}     
\def\norm #1{\Vert #1\Vert}

\def\bra #1{\langle #1\vert}
\def\ket #1{\vert #1\rangle}

\def\tr{{\rm Tr}}

\def\P{{\mathbb P}}
\def\E{{\mathbb E}}

\def\omin{{\otimes_{\text{min}}}}

\newcommand*{\half}{\frac{1}{2}}